\documentclass[10pt,journal,compsoc,twoside]{IEEEtran}

\newtheorem{theorem}{Theorem}[section]
\newtheorem{definition}{Definition}[section]
\newtheorem{example}{Example}[section]
\newtheorem{property}{Property}[section]
\newtheorem{lemma}{Lemma}[section]
\newtheorem{remark}{Remark}[section]

\newenvironment{proof}{{\textbf{\emph{Proof. }}}}{\hfill$\square\\$}

\usepackage[switch]{lineno}
\usepackage{floatrow}
\usepackage{subfig}
\usepackage[section]{algorithm}
\usepackage{algpseudocode}
\usepackage{graphicx}
\usepackage{amsmath}
\usepackage{amssymb}
\usepackage{amsfonts}
\usepackage{float}
\floatstyle{plaintop}
\restylefloat{table}
\usepackage[numbers,sort&compress]{natbib}

\usepackage{url}
\newcommand{\com}[2]{}
\newcommand{\comm}[3]{}
\newcommand{\rexp}[3]{}

\newcommand{\hr}[1]{}

\newcommand{\hrr}[1]{}
\ifCLASSINFOpdf
\else
\fi
\begin{document}
%
\title{Efficient Transition Adjacency Relation Computation for Process Model Similarity}
%
%
%
%

\author{Jisheng~Pei, Lijie~Wen, Xiaojun~Ye, and Akhil~Kumar
\IEEEcompsocitemizethanks{\IEEEcompsocthanksitem L. Wen is the corresponding author with School of Software, Tsinghua University, Beijing 100084, China.\protect\\
E-mail: wenlj@tsinghua.edu.cn
\IEEEcompsocthanksitem J. Pei, X. Ye are with Tsinghua University and A. Kumar is with Penn State University.}
}

%
%

\markboth{IEEE TRANSACTIONS ON SERVICES COMPUTING,~Vol.~x, No.~x, x~2018}%
{Pei \MakeLowercase{\textit{et al.}}: Efficient Transition Adjacency Relation Computation for Process Model Similarity}
%



\IEEEtitleabstractindextext{%
\begin{abstract}
Many activities in business process management, such as process retrieval, process mining and process integration, need to determine the similarity between business processes. Along with many other relational behavior semantics, Transition Adjacency Relation (abbr. TAR) has been proposed as a kind of behavioral gene of process models and a useful perspective for process similarity measurement. In this article we explain why it is still relevant and necessary to improve TAR or pTAR (i.e., projected TAR) computation efficiency and put forward a novel approach for TAR computation based on Petri net unfolding. This approach not only improves the efficiency of TAR computation, but also enables the long-expected combined usage of TAR and Behavior Profiles (abbr. BP) in process model similarity estimation.
\end{abstract}

}

\maketitle

\IEEEdisplaynontitleabstractindextext

%
\IEEEpeerreviewmaketitle

\section{Introduction}\label{sect:intro}

Since its first introduction \cite{TAR10}, Transition Adjacency Relation (TAR) has been applied in a wide range of business process model analysis scenarios such as similarity measurement~\cite{pTAR12,BP11,CBP10,IS11,CII12}, process model retrieval~\cite{SS09,ER12}, and process mining~\cite{DKE10,PM2nd}. Different from conformance checking between a process model and an event log~\cite{CC08}, a process similarity algorithm based on TAR (TAR similarity in short) takes two process models as the input and computes a similarity value ranging from 0.0 to 1.0 between them.  Furthermore, TAR similarity is also different from workflow inheritance techniques~\cite{IoW02}, which also takes two process models as the input but tries to determine whether one process model is inherited from another. TAR in a Petri net simply defines the set of ordered pairs $\langle t_1,t_2\rangle$ that can be executed one after another, but despite of its simplicity, the set of TARs reveal a useful perspective of the behaviors of the business process model. Later on, TAR was extended to projected Transition Adjacency Relation (pTAR) \cite{pTAR12}, which neglects the silent transitions so that users may selectively focus on a set of business-relevant transitions, and neglect the rest by marking them as silent transitions. However, both the original TAR computation algorithm \cite{TAR10} and pTAR computation algorithm \cite{pTAR12} rely on exploiting the state-space of the process model by constructing its reachability graph (RG), which leaves the state-explosion caused by concurrencies as an unsolved problem.

On the other hand, Behavior Profiles (BP) \cite{BP11} have been accepted in recent years as a most widely used relational semantics for business process analysis. It provides great descriptive power for process behaviors and can be efficiently computed with unfolding \cite{CBP10}). However, despite its various applications in scenarios such as consistency checking and similarity measurement, researchers have noticed the limitations to its expressive powers \cite{DKE10}. It has been discussed in several works that while some process models may have identical behavior profiles but produce different output traces.

As a result, a comprehensive relational semantics called 4C Spectrum has been proposed in \cite{4C14} to make up for the missing expressive power of BP. Nevertheless, although 4C Spectrum is strong in terms of expressive power, it is a heavy-weight solution and contains hundreds of complex relations (resulted from the combination of co-occurrence and causal/concurrent patterns), which is costly to compute and difficult for users to grasp the insight intuitively. As a result, \cite{4C14} has addressed the computation of only two of these relations, leaving most of the 4C Spectrum relations unavailable for application at this stage.

The above developments of relational semantics motivate us to look for a light-weight solution to significantly improve current relational semantics.  In \cite{CII12}, it has been suggested that similarity estimation might benefit a lot from the combined usage of BP and TAR. Also, it is worthy of notice that there are only 4 types of relations in total even if consider BP together with TAR. However, this idea has not been broadly applied yet because TAR computation based on reachability graphs is extremely low in efficiency, and users may not want to sacrifice the high efficiency of BP even though considering TAR  additionally may result in improvement of similarity estimation quality. To illustrate this, we will first demonstrate the advantage of TAR over BP. Motivated by this, we will go deeper into the issue of how to improve TAR computation efficiency.

\begin{figure}
  \label{fig:tarbp} 
  \centering
  \subfloat[][Parallel Structure]{
    \label{fig:tarbp:a} 
    \includegraphics[width=3.0in]{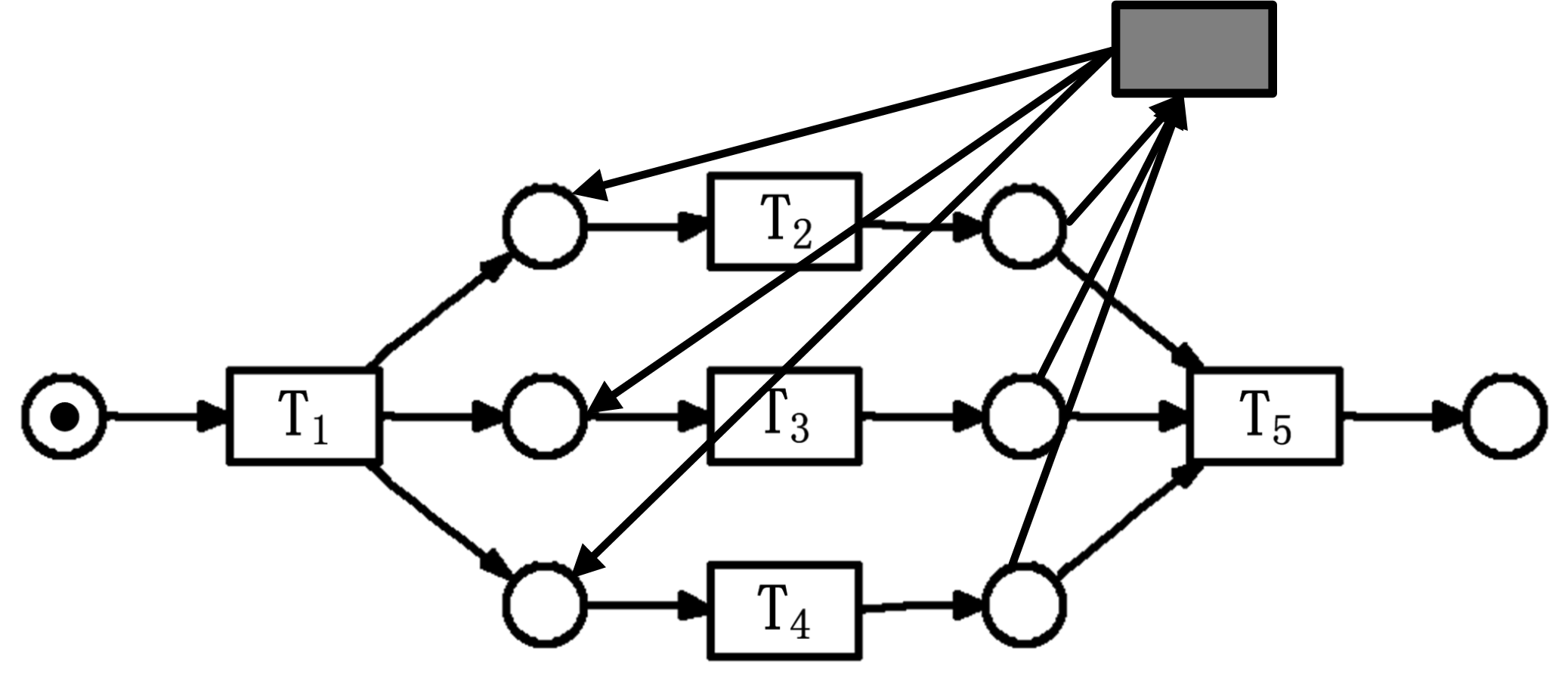}}
  \hspace{0.1in}
  \subfloat[][Loop Structure]{
    \label{fig:tarbp:b} 
    \includegraphics[width=3.4in]{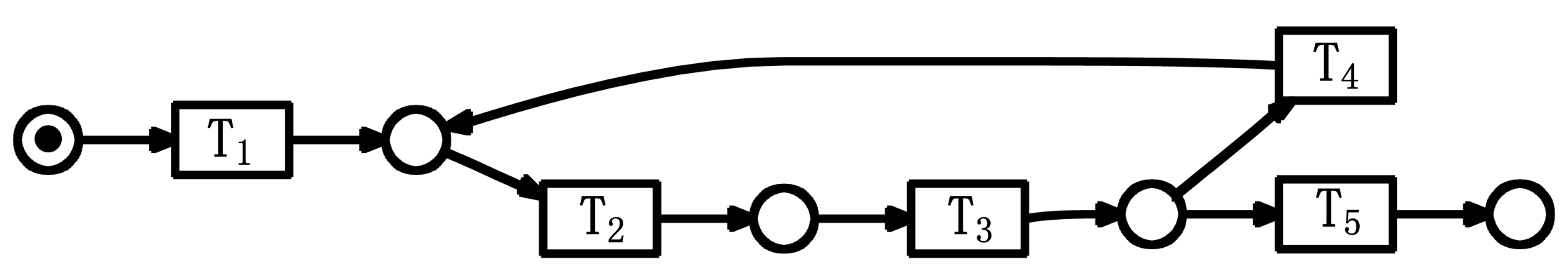}}
  \hspace{0.3in}
  \subfloat[][Sequential Structure]{
    \label{fig:tarbp:c} 
    \includegraphics[width=3.0in]{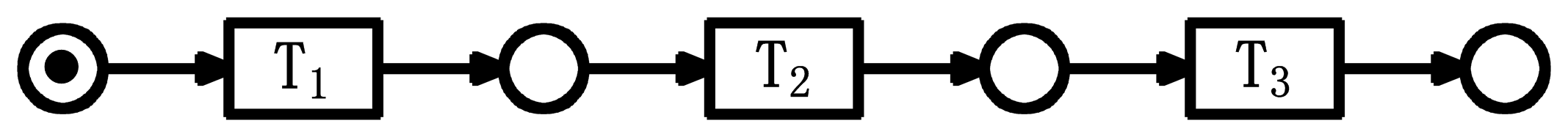}}
  \hspace{0.3in}
  \subfloat[][Sequential Structure with Skips]{
    \label{fig:tarbp:d} 
    \includegraphics[width=3.0in]{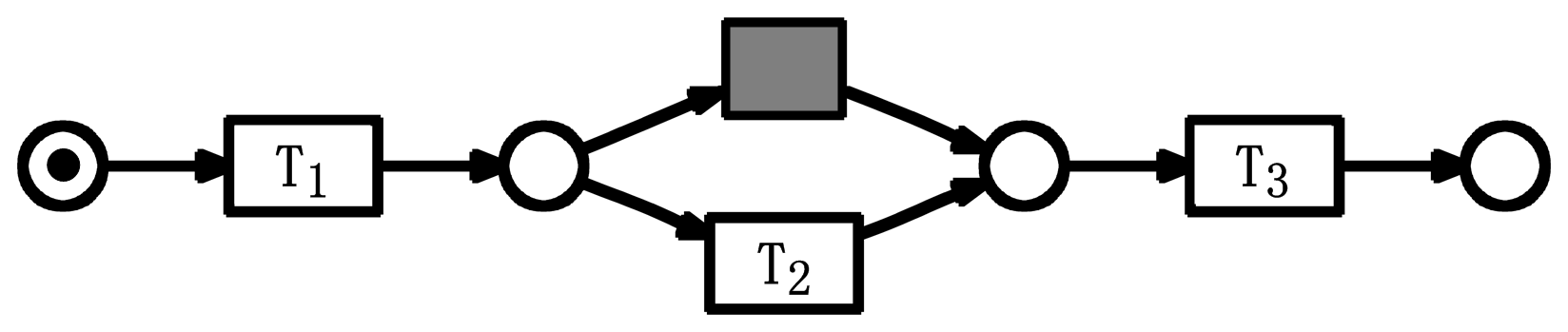}}\\
  \caption{(a) and (b) (as well as (c) and (d)) have identical BPs but different behaviors}
\end{figure}

In the above figure, the nets in Fig.~\ref{fig:tarbp:a} and Fig.~\ref{fig:tarbp:b} have identical BP (we refer readers to the original work of \cite{BP11} for definitions of BP), that is, for both nets we have $T_1\rightarrow \{T_2,T_3,T_4\}\rightarrow T_5$ and $\{T_2,T_3,T_4\}$ are mutually in $\parallel$. Also, for both nets in Fig.~\ref{fig:tarbp:c} and Fig.~\ref{fig:tarbp:d} it holds $T_1\rightarrow T_2\rightarrow T_3$. However, these nets clearly have different behaviors. $T_2,T_3,T_4$ in Fig.~\ref{fig:tarbp:a} form a parallel structure, whereas in Fig.~\ref{fig:tarbp:b} they form a loop. 

Compared to BP, which describes the global ordering between transitions, Transition Adjacency Relation (TAR) captures more `local' features among transitions that occur immediately one after another. This makes it a potential enhancement to BP based behavior description. For example, in Fig.~\ref{fig:tarbp:a}, $T_4$ may occur immediately after $T_2$ (denoted as $T_2<_{tar}T_4$), and $T_2$ may also occur immediately after $T_4$ ($T_4<_{tar}T_2$), but $T_4$ cannot occur immediately after $T_2$ in Fig.~\ref{fig:tarbp:b}. Moreover, $T_2$ in Fig.~\ref{fig:tarbp:d} can be skipped whereas $T_2$ in Fig.~\ref{fig:tarbp:c} cannot, which can be observed with TAR, as $T_3$ may occur immediately after $T_1$ ($T_1<_{tar}T_3$) in Fig.~\ref{fig:tarbp:d} but not in Fig.~\ref{fig:tarbp:c}.

As we have observed in the above examples, models with identical BP may have different TAR. Therefore, BP cannot be used to derive TAR even though unfolding-based BP computation has been addressed in \cite{CBP10}, because we cannot directly rely on BP information to compute TAR (the vice-versa is also true, so we are not arguing that TAR is a superior representation than BP, but it is a potentially useful extension to BP). Nevertheless, the CFP structure built during BP computation is already available for re-use during TAR computation.

Admittedly, apart from TAR, arbitrarily many other relational semantics can be incorporated with BP to extend its expressive power, such as the reduced event structures proposed in \cite{UNF01}. However, the trade-off between simplicity (for both computation and human understanding) and more expressive power must be taken into account. In this sense, any additional relational semantics that bring more expressive power at smaller cost of complexity is considered more desirable. Moreover, since BP are generally computed by building the complete finite prefix (CFP) of unfolding \cite{CBP10}, a preferable relational semantics for enhancement should be efficiently computable by re-using information from the CFP structure built previously.

As a consequence, we are motivated to put TAR's computation in the context of unfolding to save the marginal cost and establish the potential of TAR as such a desirable addition to BP. This article explores unfolding based TAR/pTAR computation in three stages. First, we confirm that TAR computation can generally be reduced to the classic coverability problem and therefore solvable using the `co' relations among conditions in the CFP , but the reduction of pTAR computation to coverability problems is not always trivial (Sect.~\ref{sec:General}). By re-using the causal information in the CFP more substantially, the computation of both TAR and pTAR can be further accelerated using a handful of derivation rules (Sect.~\ref{sec:More}). However, the incompleteness of unfolding information in CFP caused by cut-off events still restricts the applicability of these rules sometimes. Therefore, in Sect.~\ref{sec:ConUnf}, we describe our findings about a novel kind of finite extension of CFP, in which adequate information is preserved, so that TAR/pTAR can be derived from this extension completely using rules similar to those in Sect.~\ref{sec:More}.

An overview of our work is summarized in Fig. \ref{fig:overviewnew}. We set out to exploit all information available in CFP for TAR computation to the best of our knowledge, and achieved the leading performance among existing methods. With this work, we attempt to pave the way for further researches of the computation of TAR/pTAR and its application in business process behavior analysis, especially through combinations with other behavior abstractions such as BP.

\begin{figure*}
  \includegraphics[width=7.2in]{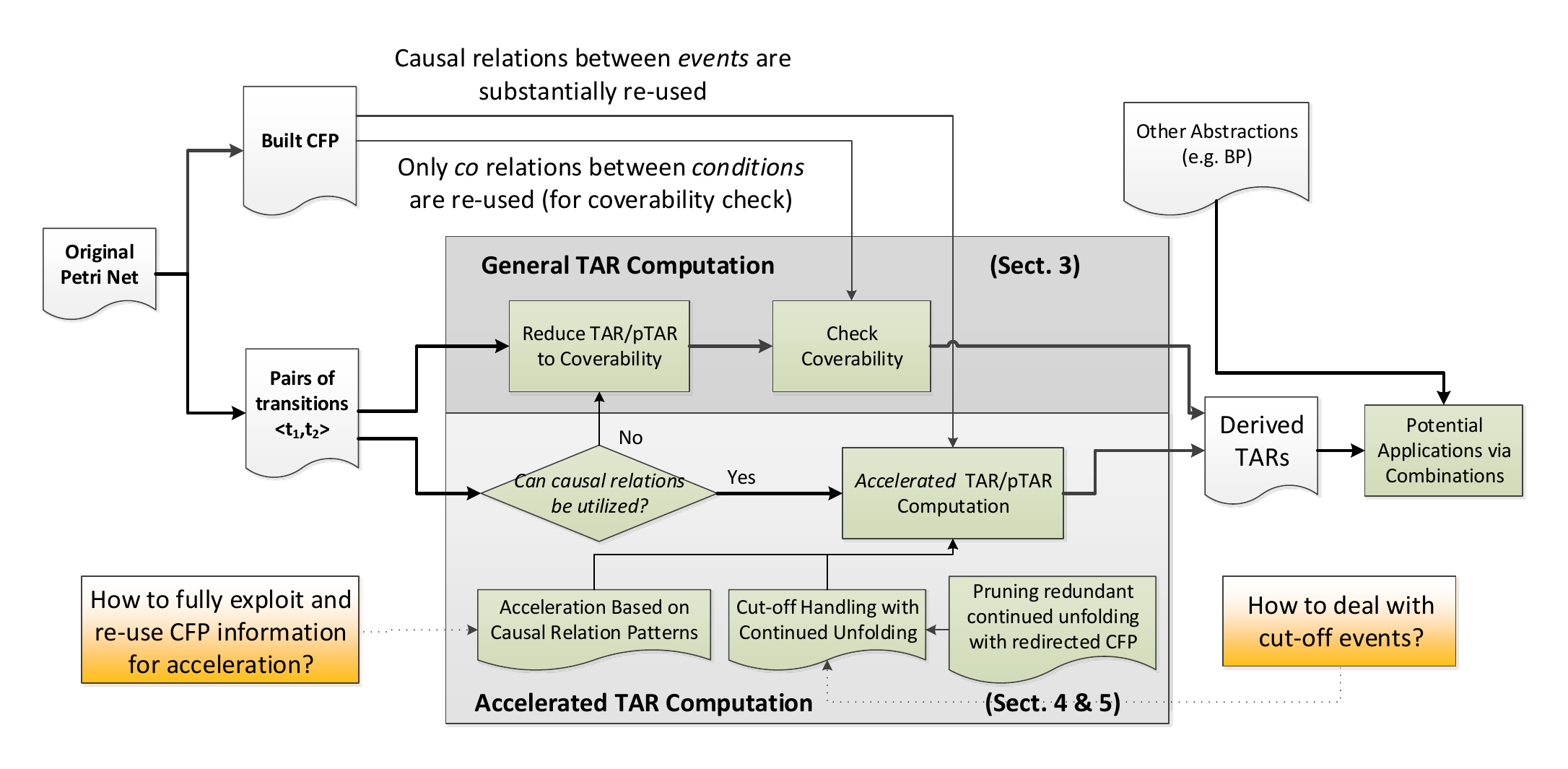}
  \caption{An overview of our research in unfolding-based TAR computation}
  \label{fig:overviewnew}
\end{figure*}

This paper is an extended version of our previous work~\cite{CTAR16}, but with significant amount of new contributions to the problem of transition adjacency relation (TAR) computation. A new type of cut-off criteria for constructing finite prefixes of Petri net unfolding has been provided in this paper to remove the necessity of costly coverability checking operation, which has been identified in our previous work~\cite{CTAR16} a major challenge to TAR computation. More in-depth proofs and experiments have been provided in this manuscript to make the proposed TAR/pTAR computation techniques more solid. Furthermore, we have provide additional business process model retrieval experiments in this version to demonstrate the necessity and application potential of combined usage of TAR and behavior profiles in this extended version.

The remaining of this paper (excepting Sect.~\ref{sec:General},~\ref{sec:More}, and~\ref{sec:ConUnf}) is organized as follows. Section~\ref{sec:Eval} evaluates the TAR/pTAR computation performance on real-world data and the scalability of our algorithms. In Section~\ref{sec:Comb}, we show the effectiveness of the combination of TAR and BP by a concrete experiment on process model querying. Section~\ref{sec:Conc} concludes our work.

\section{Preliminary Concepts}\label{sec:PreCon}

We first recall the basic definitions on Petri net syntax and semantics \cite{PN89}.

\begin{definition}[Petri net]\label{def:PetriNet}
A Petri net is a triple $N=(P,T,F)$, where P and T are finite set of places and transitions respectively $(P\cap
T=\emptyset$ and $P\cup T\neq \emptyset)$, and $F\subseteq (P\times T)\cup(T\times P)$ is a set of arcs (flow relation). Denote $X=P\cup T$, for a node $x\in X$, $\bullet x=\{y\in X|(y,x)\in F\}$, $x\bullet=\{y\in X|(x,y)\in F\}$.
\end{definition}

\begin{definition}[Petri net semantics]\label{def:FireRule}
Let $N=(P,T,F)$ be a Petri net, then
\begin{itemize}
\item $M:P\rightarrow \mathbb{N}$ is a marking of $N$, $\mathbb{N}$ is the set of non-negative integers.
$\mathbb{M}$ denotes all possible markings of $N$. $M(p)$ denotes the number of tokens in place $p$.
\item A transition t is enabled in $M$, iff $\forall p\in\bullet t: M(p)\geq 1$. If t is enabled at $M$ then it can be fired, and then a new marking $M'$ is reachable from $M$ by firing t, denoted by $M\xrightarrow{t}M'$, such that $M'=(M\backslash\bullet t)\cup t\bullet$.
\item A net system is a pair $S=(N,M_{0})$, where $N$ is a net and $M_{0}$ is the initial marking of $N$.
\item A firing sequence $\sigma=\langle t_{1},...,t_{n-1}\rangle$ leads $N$ from marking $M_{1}$ to marking $M_{n}$  such that $M_{1}\xrightarrow{t_{1}}M_{2}\xrightarrow{t_{2}}M_{3}...M_{n-1}\xrightarrow{t_{n-1}}M_{n}$, also denoted as $M_{1}\xrightarrow{\sigma}M_{n}$. We use $\sigma(i)$ to denote the $i^{th}$ transition in sequence $\sigma$.
\item For a net system $S(N,M_0)$, $M$ is a reachable marking if from initial marking $M_0$ there exists a firing sequence $\sigma$ such that $M_{0}\xrightarrow{\sigma}M$.
\item A net system is \textbf{bounded} if it has a finite number of reachable markings.
\item $S(N,M_0)$ is \textbf{1-safe} iff for each reachable marking $M$, $M(p)\leq$1 for every place p, and is \textbf{free-choice} iff $\forall t_1,t_2\in T:\bullet t_1\cap\bullet t_2\neq\emptyset\Rightarrow \bullet t_1=\bullet t_2$.
\end{itemize}
\end{definition}

The unfolding of a Petri net is derived from its related Occurrence Net, whose definition is based on the concepts of \emph{causal}, \emph{conflict}, and \emph{concurrency} relations between Petri net nodes \cite{CPU02}. In a Petri net, two nodes $x$ and $y$ are in \emph{causal} relation, denoted by $x < y$, if the net contains a path with at least one arc leading from $x$ to $y$. Furthermore, $x$ and $y$ are in \emph{conflict} relation, denoted by $x\# y$, if the net contains two paths $s\cdot t_1...\cdot x_1$ and $s\cdot t_2...\cdot x_2$ starting at the same place $s$, and such that $t_1\neq t_2$. Finally, $x$ and $y$ are in \emph{concurrency} relation, denoted by $x$ \emph{co} $y$, if neither $x < y$ nor $y < x$ nor $x\# y$.
\begin{definition}[Occurrence net]\label{def:OccurrenceNet}
An occurrence net is a net $O=(C,E,G)$, in which places are called conditions ($C$), transitions are called events ($E$), and:
\begin{itemize}
\item $O$ is acyclic, or equivalently, the causal relation is a partial order.
\item $\forall c\in C:|\bullet c| \leq 1$, and for all $x\in C\cup E$ it holds $\neg(x\#x)$ and the set $\{y\in C\cup E|y<x\}$ is finite.
\item For an occurrence net $O=(C, E, G)$, $Min(O)$ denotes the set of minimal elements of $C\cup E$ with respect to $<$.
\end{itemize}
\end{definition}

The relation between a net system $S=(N,M_{0})$ with $N=(P,T,F)$ and an occurrence net $O=(C,E,G)$ is defined as a homomorphism
$h:C\cup E\longmapsto P\cup T$ such that $h(C)\subseteq P$ and $h(E)\subseteq T$.

A \emph{branching process} of $S = (N,M_0)$ is a tuple $\pi= (O, h)$ with $O=(C, E, G)$ being an occurrence net and $h$ being a homomorphism from $O$ to $S$. The maximal branching process of $S$ is called unfolding \cite{CPU02}. The unfolding of a net system can be truncated once all markings of the original net system and all enabled transitions are represented. This yields the complete finite prefix (abbr. CFP).

\begin{definition}[Complete Finite Prefix]\label{def:CPU}
Let $S=(N,M_0)$ be a system and $\pi=(O,h)$ a branching process with $N=(P,T,F)$ and $O=(C,E,G)$.
\begin{itemize}
\item A set of events $E'\subseteq E$ is a configuration, iff $\forall e, f\in E': \neg(e\#f)$ and $\forall e\in E': f<e \Rightarrow f\in E'$. The local configuration $[e]$ for an event $e\in E$ is defined as $\{x\in E|x<e\vee x=e\}$.
\item A set of conditions $X\subseteq C$ is called \emph{co}-set, iff for all distinct $c_1, c_2\in X$ it holds $c_1$ co $c_2$. If $X'$ is maximal w.r.t. set inclusion, it is called a cut. For a finite configuration $E'$, $Cut(E')=(Min(O)\cup E'\bullet)\backslash\bullet E'$ is a cut, while $h(Cut(E'))$ is a reachable marking of $S$, denoted as $Mark(E')$, e.g. $Mark([e])$ is the set of places that are marked after all events in $[e]$ are fired.
\item An adequate order $\prec$ is a strict well-founded partial order on local configurations such that for two events $e, f\in E$, $[e]\subset[f]$ implies $[e]\prec[f]$.
\item An event $e$ is a cut-off event if there exists another event $e'$ such that $Mark([e])=Mark([e'])$ and $[e']\prec[e]$. $e'$ is the corresponding event of $e$, denoted as $e'=corr(e)$.
\item A complete finite prefix (CFP) is the greatest backward-closed subnet of a branching process containing no events after any cut-off event.
\end{itemize}
\end{definition}
\begin{example}

\begin{figure}
  \centering
  \subfloat[][Original Net]{
    \label{fig:explod:a} 
    \includegraphics[width=3in]{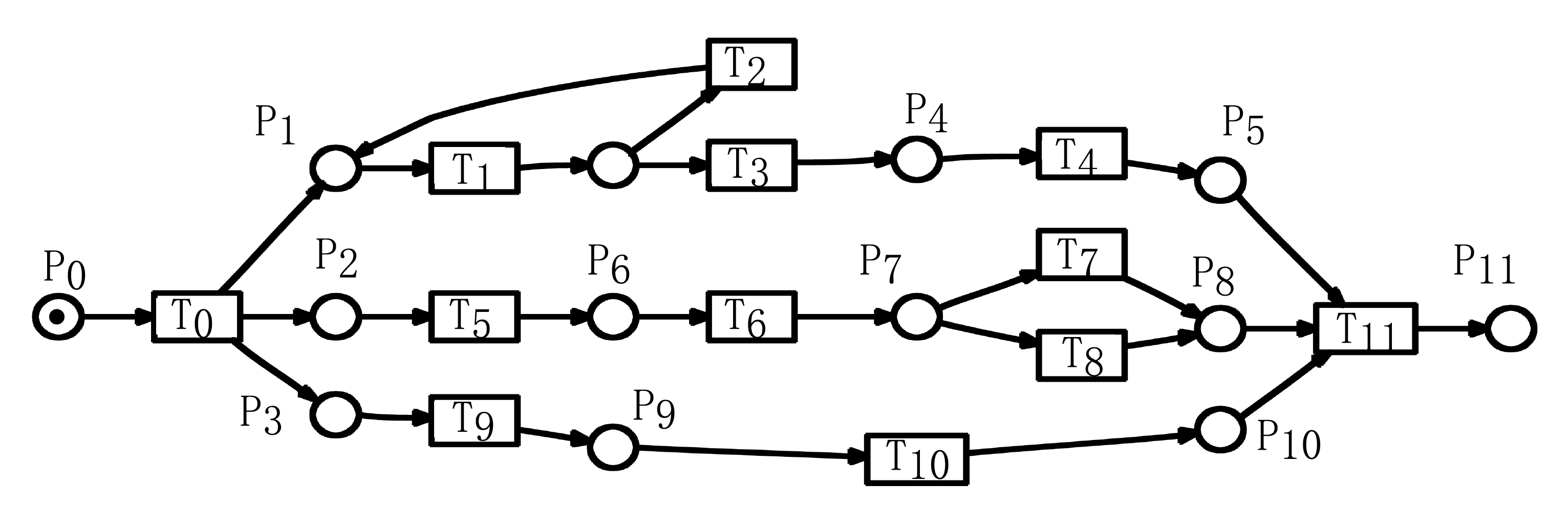}}
  \hspace{0.1in}
  \subfloat[][Reachability Graph]{
    \label{fig:explod:c} 
    \includegraphics[width=2.3in]{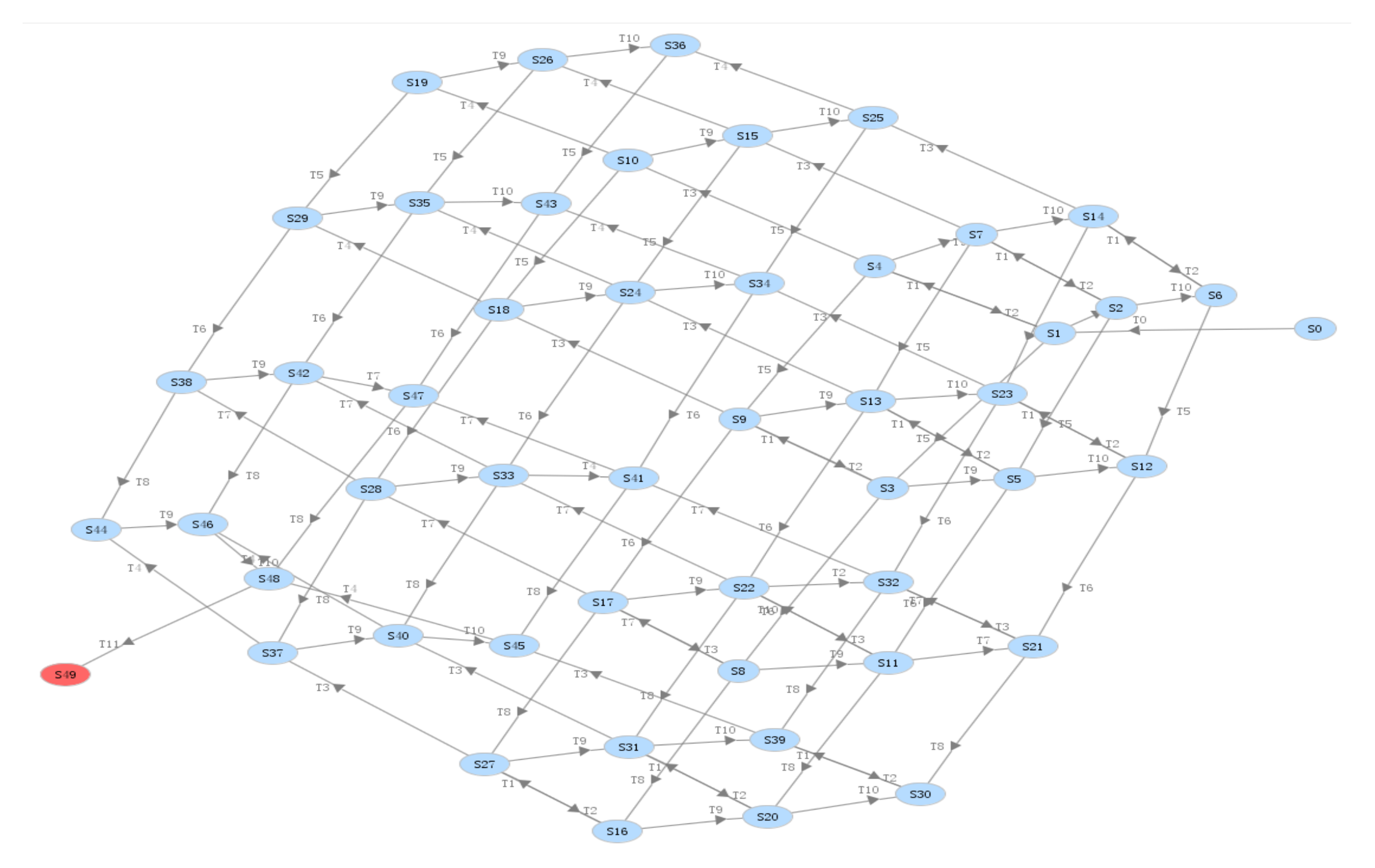}}
  \hspace{0.3in}
  \subfloat[][Corresponding CFP of \textbf{(a)}]{
    \label{fig:explod:b} 
    \includegraphics[width=3.3in]{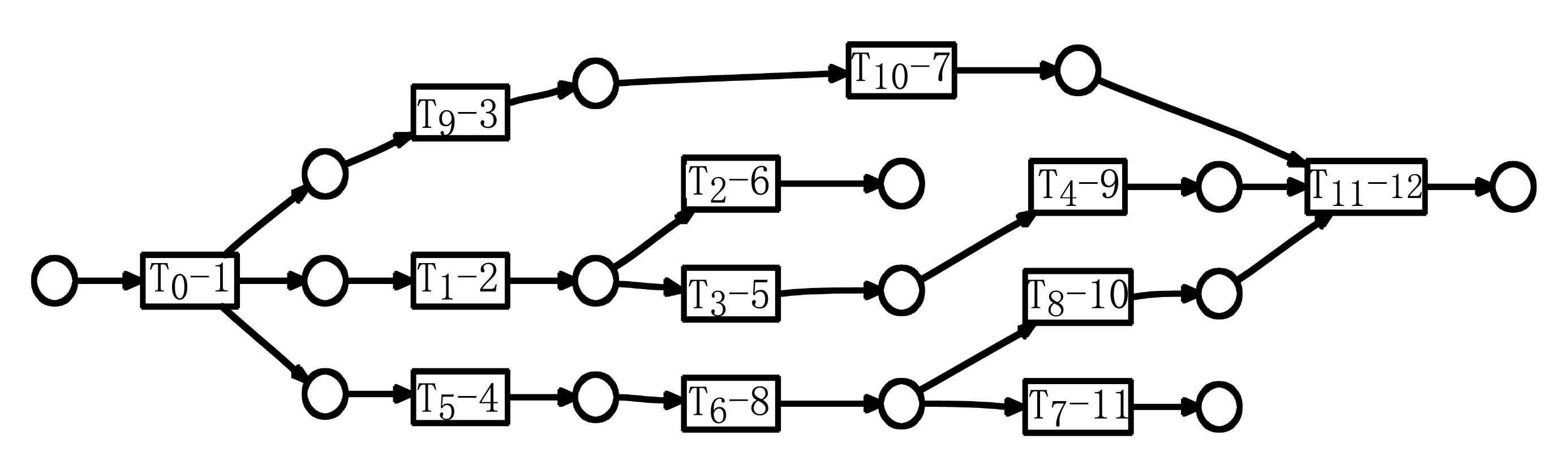}}

  \label{fig:explod} 
  \caption{CFP as a more compact representation for Petri net behavior analysis}
\end{figure}

Fig.~\ref{fig:explod:a}, \ref{fig:explod:c} and \ref{fig:explod:b} illustrate a Petri net, its related reachability graph and  CFP respectively (Fig.~\ref{fig:explod:c} is not supposed to be readable, only intended to illustrate the complexity of the reachability graph). Although there are multiple solutions as Fig.~\ref{fig:explod:a}'s CFP, we only show one of them here for simplicity. In this paper, we name each event with its related transition name, \textbf{which is followed by `$\mbox{-}$' and then its order of appearance in the CFP.} For example, in Fig.~\ref{fig:explod:b}, $T_8\mbox{-}10$ is an event that maps to $T_8$ of the original net ($h(T_8\mbox{-}10)=T_8$), and it is the tenth event constructed in the CFP. In Fig.~\ref{fig:explod:b}, $[T_8\mbox{-}10]=\{T_0\mbox{-}1,T_5\mbox{-}4,T_6\mbox{-}8,T_8\mbox{-}10\}$, $[T_7\mbox{-}11]=\{T_0\mbox{-}1,T_5\mbox{-}4,T_6\mbox{-}8,T_7\mbox{-}11\}$, and $Mark([T_8\mbox{-}10])$$=Mark([T_7\mbox{-}11])$$=\{P_{1},P_{3},P_{8}\}$, thus $T_7\mbox{-}11$ is a \emph{cut-off} event, and $T_8\mbox{-}10$ is its \emph{corresponding event}. $T_2\mbox{-}6$ is a \emph{cut-off} event caused by looping, whose corresponding event is $T_0\mbox{-}1$.
\end{example}

A reachability graph (RG) is prone to encounter state-explosions with concurrent behaviors (because it enumerates all possible execution traces), whereas the CFP of a Petri net only `outlines' the net's behavior and is much more compact, as illustrated in Fig. \ref{fig:explod:c} and \ref{fig:explod:b}.

In this paper, we adopt Esparza's algorithm \cite{CPU02} for CFP construction, which can generate more compact CFPs compared to the original CFP construction technique. We also refer readers to \cite{CPU02} for more examples on unfolding.

Let us now present the definition of Transition Adjacency Relation (TAR) and Projected Transition Adjacency Relation (pTAR):

\begin{definition}[Transition Adjacency Relation]\label{def:TAR}
Let $S=(N,M_0)$ be a net system. Let $t_1, t_2$ be two transitions of S. We say that $t_1, t_2$ are in Transition Adjacency Relation, denoted as $t_1 <_{tar} t_2$, if there exist a reachable marking $M_s$ of $S$, where $t_1$ is enabled, and $M_s\xrightarrow{t_1}M_s'$, such that $t_2$ is enabled at $M_s'$.
\end{definition}

\begin{definition}[Projected Transition Adjacency Relation]\label{def:STAR} Let $t_1,t_2$ be two non-silent transitions of a net system $S(N,M_0)$. We say that $t_1$ and $t_2$ are in projected Transition Adjacency Relation (pTAR), denoted as $t_1<_{ptar}t_2$ iff from $M_0$ a sequence $\sigma$ can be fired such that $\sigma(i)=t_1,\sigma(j)=t_2$ for some $1\leq i<j$ and $\sigma(k)$ are silent transitions for all $i<k<j$.
\end{definition}

\begin{figure}
  \caption{Illustration of TAR and pTAR}
  \includegraphics[width=3.2in]{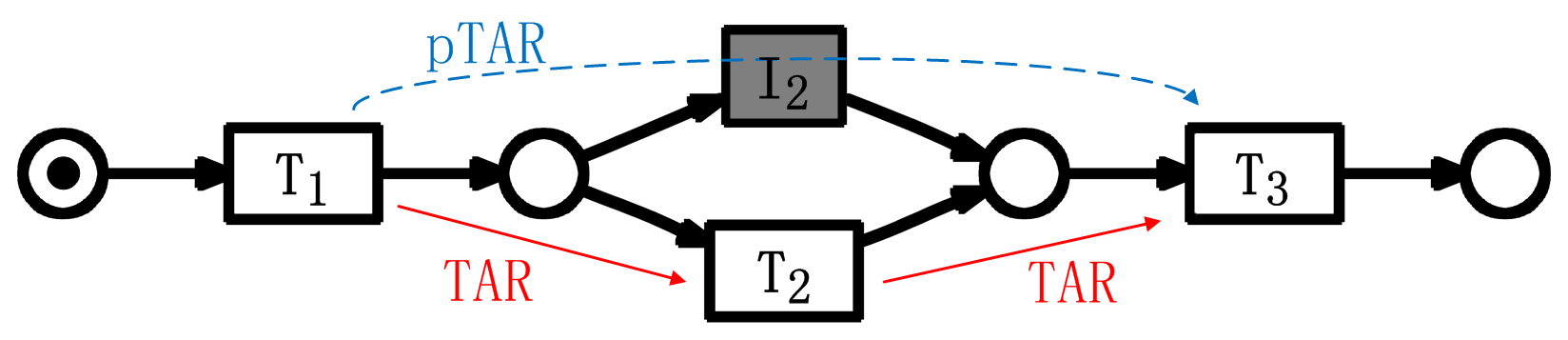}
  \label{fig:TARpTAR}
\end{figure}

Intuitively, TAR says `some pair of transitions can occur immediately one after another', and pTAR (Projected TAR, or Silent TAR) says `some pair of transitions can occur immediately one after another, if the silent transitions occurring between them are neglected'. For example, in Fig.~\ref{fig:TARpTAR}, $T_1$ and $T_2$, $T_2$ and $T_3$ are in TAR, denoted as $T_1<_{tar}T_2,T_2<_{tar}T_3$, and $T_1$ are in pTAR with $T_3$ ($T_1<_{ptar}T_3$), as we can skip $T_2$ by firing the silent transition $I_2$. The definition is equivalent with the the one in \cite{pTAR12}, where projected transitions are treated as non-silent transitions and the rest as silent transitions.



\section{Baseline Approach of TAR Computation from CFP}\label{sec:General}

The CFP of a Petri net contains various information that we can exploit for TAR computation. In this section, we look at how to exploit the $co$ relations between \emph{conditions} in CFP, by reducing TAR computation to coverability checking.

\subsection{Reduction of TAR Computation to Coverability Problem}

First we show that the computation of TAR can be reduced to the classic coverability problem, that is, checking if a certain set of places can be `covered' by some reachable marking of a Petri net. We prove this as follows.

\begin{theorem}\label{thm:coverable} For two transitions $t_1,t_2$ in a bounded net system: $t_1<_{tar}t_2\Leftrightarrow (\bullet t_2\backslash t_1\bullet)\cup\bullet t_1$ is coverable.
\end{theorem}

\begin{proof}
By definition of TAR, we have $t_1<_{tar}t_2$ iff there exists a reachable marking $M: M\xrightarrow{t_1}M'\xrightarrow{t_2}M''$. Then $M\xrightarrow{t_1}M'\xrightarrow{t_2}M''$ holds if and only $\bullet t_2\subseteq M'$ and $\bullet t_1\subseteq M$, which holds if and only if $(\bullet t_2\backslash t_1\bullet)\cup\bullet t_1\subseteq M$.

Since $M$ is a reachable marking of $S$, by the above reasoning we have that $t_1<_{tar}t_2\Leftrightarrow (\bullet t_2\backslash t_1\bullet)\cup\bullet t_1$ is coverable.
\end{proof}

In CFP, a coverable marking of the net system corresponds with a set of conditions that are mutually in $co$ relations (i.e., a $co$-set) \cite{UNF01}. Consequently, the coverability problem can be solved efficiently based on the $co$ relations between conditions in CFPs using SAT solvers \cite{UNF95}. Therefore, a general algorithm of unfolding-based TAR computation can be given as Algorithm~\ref{alg:basic}.

\begin{algorithm}
  \caption{General TAR Derivation Using CFP (GENERAL)}\label{alg:basic}
  \begin{algorithmic}[1]
    \Function {GENERAL}{$S,\pi$}
        \State $TAR\leftarrow\emptyset$
        \State \textbf{for}{\textbf{ each} $t_1,t_2\in S$} \textbf{do}
            \State \indent Check coverability of $(\bullet t_2\backslash t_1\bullet)\cup\bullet t_1$ based on $\pi$ using SAT solvers
            \State \indent\textbf{if} $(\bullet t_2\backslash t_1\bullet)\cup\bullet t_1$ is coverable \textbf{then} $TAR\leftarrow TAR\cup \{\langle t_1,t_2\rangle\}$
       \State return $TAR$
    \EndFunction
  \end{algorithmic}
\end{algorithm}

\subsection{Reduction of pTAR to Coverability And Its Non-triviality}

Unlike the original TAR, the checking of projected TARs (pTARs) between a given pair of transitions $t_1,t_2$ cannot be trivially reduced to coverability problem. Unlike in Thm. \ref{thm:coverable}, where $t_2$ can be fired immediately after $t_1$, it is probable that after firing $t_1$, a sequence of silent transitions have to be fired before $t_2$ is fired. In this case, suppose we want to reduce the checking of pTAR to the coverability problem, we need to know the firing sequence between $t_1$ and $t_2$. In other words, we need to determine a sequence $\sigma=\langle t_{i_1},...,t_{i_k}\rangle$ such that there exists a reachable marking $M$ such that $M\xrightarrow{t_1}M_1\xrightarrow{t_{i_1}}M_2\xrightarrow{t_{i_2}}...M_k\xrightarrow{t_{i_k}}M_{k+1}\xrightarrow{t_2}M_{k+2}$. By the same reasoning in Thm. \ref{thm:coverable}, $P_k=(\bullet t_2\backslash t_{i_k}\bullet)\cup\bullet t_{i_k}$ must be covered by $M_{k}$, $P_{k-1}=(P_k\backslash t_{i_{k-1}}\bullet)\cup\bullet t_{i_{k-1}}$ must be covered by $M_{k-1}$, etc. So the pTAR checking between $t_1$ and $t_2$, with respect to some $\sigma$, can be reduced to the coverability checking of $P_1=(...(((\bullet t_2\backslash t_{i_k}\bullet)\cup\bullet t_{i_k})\backslash t_{i_{k-1}}\bullet)\cup\bullet t_{i_{k-1}}...\backslash t_1\bullet)\cup\bullet t_1$.

However, although pTAR checking can be reduced to coverability problem in this way, the determination of a concrete $\sigma$, which is the prerequisite for such reduction, is non-trivial, because there can be numerous possible firing paths between a pair of transitions.

\section{Exploiting More CFP Relations for Accelerations}\label{sec:More}

Although being complete and correct, Algorithm~\ref{alg:basic} includes the costly operation of coverability checking, and pTAR checking can not be reduced to coverability problem at trivial costs.

In Algorithm~\ref{alg:basic}, since TAR computation is reduced to coverability checking, only $co$ relations between conditions are used, and the rest of the CFP (causal relations and $co$ between events) are not utilized. We believe these unused information will be exactly the keys to accelerate TAR computation. By dividing the TAR cases into the following two categories, we may construct more connections between TAR and information in CFP:\\
\indent\indent\indent$(1)~t_1\bullet\cap\bullet t_2 = \emptyset$ (\textbf{non-consecutive}),\\
\indent\indent\indent$(2)~t_1\bullet\cap\bullet t_2\neq\emptyset$ (\textbf{consecutive}).

In case of $(1) t_1\bullet\cap\bullet t_2 = \emptyset$, since $t_1$ and $t_2$ in this case are non-consecutive, we have the intuition that TARs between non-consecutive transitions corresponds to $co$ relations in CFP. Fortunately, we find that such TARs can always be identified in CFP as $co$ relations because of the following useful property of CFP.

\begin{property}\label{prop:first} For every reachable marking $M$ of a Petri net system $S$, there exists a configuration $\mathcal{C}$ in a corresponding CFP $\pi$ of $S$ such that $Mark(\mathcal{C})=M$ and $\mathcal{C}$ does not contain any cut-off event.
\end{property}
\begin{proof}
Given a configuration $\mathcal{C}$, we denote by $\mathcal{C} \oplus E$ the fact that $\mathcal{C} \cup E$ is a
configuration such that $\mathcal{C} \cap E = \emptyset$. We say that $\mathcal{C} \oplus E$ is an extension of $\mathcal{C}$, and that $E$ is a suffix of $\mathcal{C}$. Let $\mathcal{C}_1$ be a configuration in the unfolding of $S$ such that Mark($\mathcal{C}_1$)=$M$. Suppose $\mathcal{C}_1$ contains some cut-off event e1 and that $\mathcal{C}_1$ = [$e_1$]$\oplus E_1$ for some set of events $E_1$. By the definition of a cut-off event, there exists a local configuration [$e_2$] such that [$e_2$] $\prec$ [$e_1$] and Mark([$e_2$]) = Mark([$e_1$]).According to the original theory of Petri net unfolding, there is a set of events isomorphic (corresponding to the same set of transitions) with $E_1$, denoted as $I_1^2(E_1)$, such that Mark([$e_2$]$\oplus I_1^2(E_1))=$Mark($\mathcal{C}_1$). Consider the configuration $\mathcal{C}_2$ = [$e_2$]$\oplus I_1^2(E_1)$. Since $\prec$ is preserved by finite extensions, we have $\mathcal{C}_2\prec \mathcal{C}_1$. If $\mathcal{C}_2$ still contains any cut-off event, we can repeat the procedure and find a configuration $\mathcal{C}_3$ such that $\mathcal{C}_3 \prec \mathcal{C}_2$ and that Mark($\mathcal{C}_3$)=Mark($\mathcal{C}_2$). The procedure cannot be iterated infinitely often because $\prec$ is well-founded. Therefore, it terminates in a configuration $\mathcal{C}$ in CFP which does not contain any cut-off event and Mark($\mathcal{C}$) = $M$.
\end{proof}

\begin{theorem}\label{prop:co}
Let $S$ be a bounded net system and $\pi=(O,h)$ its CFP, where $O=(C,E,G)$ . Let $t_1,t_2$ be two transitions of $S$. When $t_1\bullet\cap\bullet t_2=\emptyset$, we have $t_1<_{tar}t_2\Leftrightarrow\exists e_1,e_2\in E:e_1$ co $e_2$, $h(e_1)=t_1, h(e_2)=t_2$.
\end{theorem}

\begin{proof}
($\Rightarrow$) By $t_1<_{tar}t_2$, there are markings $M_s, M_s'$ in $S$ such that $M_s\xrightarrow{t_1}M_s'$, and that $t_2$ is enabled at $M_s'$. Since $t_1\bullet\cap\bullet t_2=\emptyset$, we know both $\bullet t_1$ and $\bullet t_2$ have tokens in them at $M_s$. By Property \ref{prop:first}, there exists a configuration $C$ in $\pi$ such that $C$ does not contain any cut-off event and that $Mark(C)=M_s$. Therefore $\exists X_1\subseteq Cut(C):h(X_1)=\bullet t_1$. Since $C$ does not contain cut-off events, there exists in $\pi$ an event $e_1:h(e_1)=t_1, \bullet e_1=X_1$. Similarly, because $t_2$ is enabled at $M_s'$, $\exists e_2\in\pi:h(e_2)=t_2,\bullet e_2=X_2\subseteq Cut(C)$ and $X_1\cap X_2=\emptyset$. Therefore, it holds that $e_1\not <e_2$ and $\neg(e_1\#e_2)$, and consequently, $e_1$ $co$ $e_2$. ($\Leftarrow$) For some $e_1,e_2\in E:e_1$ co $e_2$, $h(e_1)=t_1, h(e_2)=t_2$, we know both $\bullet e_1$ and $\bullet e_2$ are co-set and $\forall c_1\in\bullet e_1, c_2\in\bullet e_2: c_1$ co $c_2$. Therefore $\bullet e_1\cap\bullet e_2=\emptyset$ and $\bullet e_1\cup\bullet e_2$ is a co-set, so there exists a configuration $C$ in $\pi$ such that $\bullet e_1\cup\bullet e_2\subseteq Cut(C)$. As $h(\bullet e_1)=\bullet t_1$ and $h(\bullet e_2)=\bullet t_2$, both $t_1, t_2$ are enabled at $Mark(C)$. And after $t_1$ is fired, $t_2$ is still enabled (because $\bullet e_1\cap\bullet e_2=\emptyset$). Therefore we have $t_1<_{tar}t_2$.
\end{proof}

Now we consider the case of $(2)t_1\bullet\cap\bullet t_2\neq\emptyset$. To solve this case, we introduce the concept of Max-Event Set.

\begin{definition}
[Max-Event Set] Let $E_0$ be a subset of the events in a CFP. We call the set of maximal events in $E_0$ with respect to causal relation as its max-event set, denoted as $Max(E_0)=\{e\in E_0|\forall e' \in E_0: e\not< e'\}$.
\end{definition}

Given this definition, we have discovered the following results (see \cite{CTAR17}).

\begin{lemma}\label{lemm:first}
Let $\mathcal{C}$ be a configuration. Let $Max(\mathcal{C})=\{e_{m_1},...,e_{m_k}\}$. Then
\begin{itemize}
\item[(1)] $\mathcal{C}\backslash\{e\}$ is a configuration if and only if $e\in Max(\mathcal{C})$.
\item[(2)] $[e_{m_1}]\cup...\cup[e_{m_k}]=\mathcal{C}$.
\end{itemize}
\end{lemma}

\begin{lemma}\label{lemm:second}
Let $X$ be a \emph{co}-set in a CFP. Let $Max(\bullet X)=\{e_{m_1},...,e_{m_k}\}$. Then $\mathcal{C}=[e_{m_1}]\cup...\cup[e_{m_k}]$ is a configuration, and $X\subseteq Cut(\mathcal{C})$, $Max(\mathcal{C})=Max(\bullet X)$.
\end{lemma}

\begin{lemma}\label{lemm:third} $e\in Max(\mathcal{C})\Rightarrow\bullet e\subseteq Cut(\mathcal{C}\backslash\{e\})$.
\end{lemma}

In case (2), when $t_1\bullet\cap \bullet t_2\neq\emptyset$, we capture the intuition of `non-blocking' by introducing the following concept of Max-Event Adjacent (MEA).

\begin{definition}[Max-Event Adjacent]\label{def:madj}
For two events $e_1$ and $e_2$ in a CFP such that $e_1<e_2$, $e_1$ is said to be Max-Event Adjacent (MEA) with $e_2$, denoted as $e_1\triangleright e_2$, iff $e_1\in Max(\bullet(\bullet e_2))$. $e_1\triangleright e_2$ can be equivalently defined as $\forall c\in e_1\bullet:c<e_2\Rightarrow c\in\bullet e_2$ (more convenient and efficient for implementation).
\end{definition}

Intuitively, $e_1\triangleright e_2$ means $e_1$ is directly connected to $e_2$, but with an additional restriction that no third event `blocks' any path between them. For example, in Fig. \ref{fig:maxadj:b}, which shows the CFP of the net in Fig. \ref{fig:maxadj:a}, we have $\bullet(\bullet T_3\mbox{-}5)=\{T_0\mbox{-}1,T_2\mbox{-}3\}$, but only $T_2\mbox{-}3$ is \emph{MEA} with $T_3\mbox{-}5$. $T_0\mbox{-}1$ is \emph{not} \emph{MEA} with $T_3\mbox{-}5$ because $T_0\mbox{-}1<T_2\mbox{-}3$ and consequently $T_0\mbox{-}1\not\in Max(\bullet(\bullet T_3\mbox{-}5$)) (intuitively, $T_2\mbox{-}3$ `blocks' one of the paths between $T_0\mbox{-}1$ and $T_3\mbox{-}5$). Likewise, we have $T_2\mbox{-}4\triangleright T_4\mbox{-}6$ and $\neg(T_1\mbox{-}2\triangleright T_4\mbox{-}6)$.

MEA ($\triangleright$) is a sufficient condition to detect TARs in case (2), as we can prove using Lemma \ref{lemm:first} and Lemma \ref{lemm:third} that for any bounded Petri net's CFP:

\begin{theorem}\label{prop:direct}
        $e_1\triangleright e_2\Rightarrow h(e_1)<_{tar}h(e_2)$.
\end{theorem}

\begin{proof}
Let $\mathcal{C}=[e_2]$. By definition of local configuration, $Max([e_2])=\{e_2\}$. Therefore, by Lemma \ref{lemm:first}, $\mathcal{C'}=\mathcal{C}\backslash \{e_2\}$ is a configuration, and by Lemma \ref{lemm:third} we have: (a) $\bullet e_2\subseteq Cut(\mathcal{C'}).$

Moreover, it is easy to prove that $Max([e_2]\backslash \{e_2\})\subseteq \bullet(\bullet e_2)$. From $e_1\triangleright e_2$, we know $e_1\in Max(\bullet(\bullet e_2))$, and therefore $e_1\in Max([e_2]\backslash \{e_2\})=Max(\mathcal{C'})$. Consequently $\mathcal{C''}=\mathcal{C'}\backslash\{e_1\}$ is also a configuration. And by Lemma \ref{lemm:third} it holds that: (b) $\bullet e_1\subseteq Cut(\mathcal{C''}).$

Let the marking of $\pi$ corresponding to $Cut(\mathcal{C''}),Cut(\mathcal{C'})$ and $Cut(\mathcal{C})$ be denoted as $M_1,M_1',M_2$ respectively. By (a) and (b), it holds that $M_1\xrightarrow{e_1}M_1'\xrightarrow{e_2}M_2$ and therefore $h(e_1)<_{tar}h(e_2)$.
\end{proof}

\begin{figure}
  \centering
  \subfloat[][$T_0,T_1$ in TAR with $T_2$ but not $T_3,T_4$]{
    \label{fig:maxadj:a} 
    \includegraphics[width=2.7in]{cpulizi.pdf}}
  \hspace{0.1in}
  \subfloat[][An intuitive illustration of Max-Event Adjacent structure]{
    \label{fig:maxadj:b} 
    \includegraphics[width=3.0in]{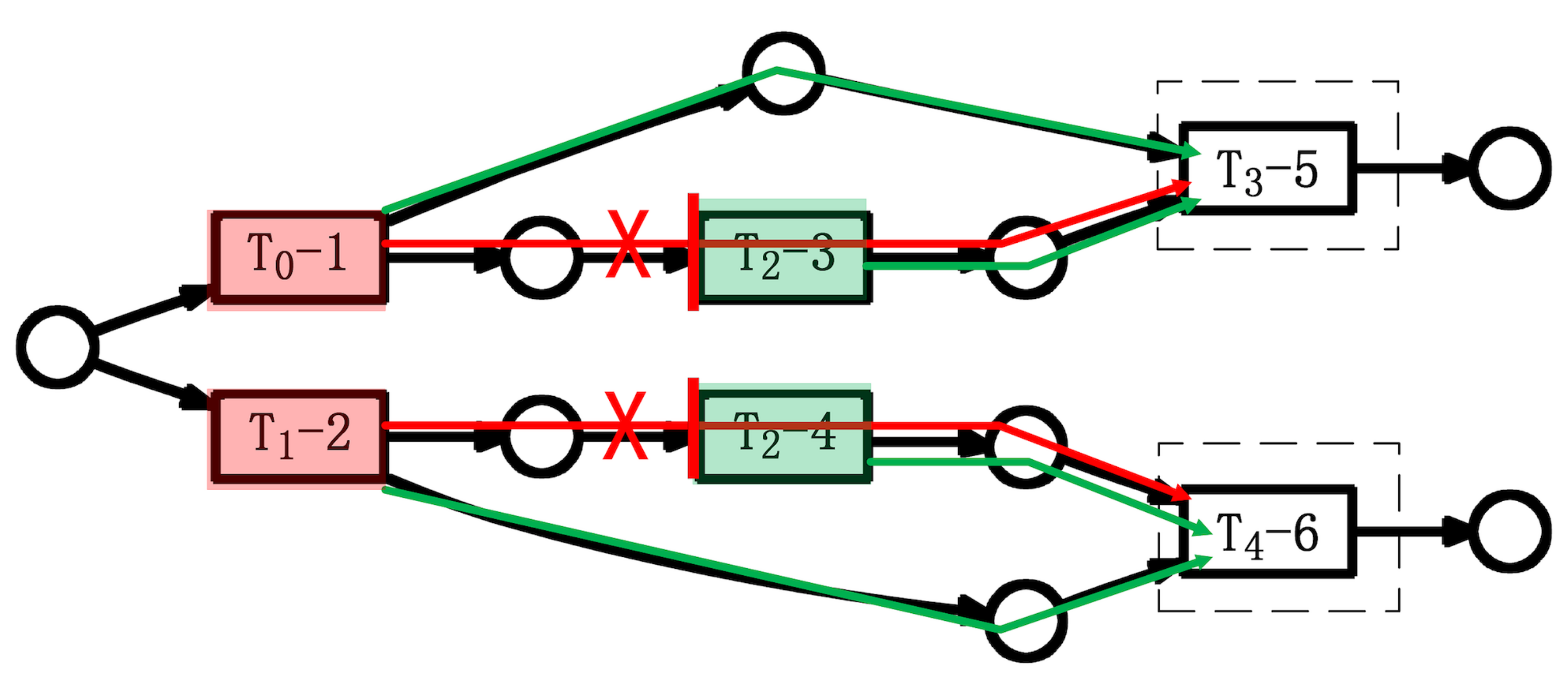}}
  \hspace{0.3in}
  \subfloat[][$T_2$ and $T_4$ are both in TAR with $T_5$]{
    \label{fig:subfigjump:a} 
    \includegraphics[width=3.4in]{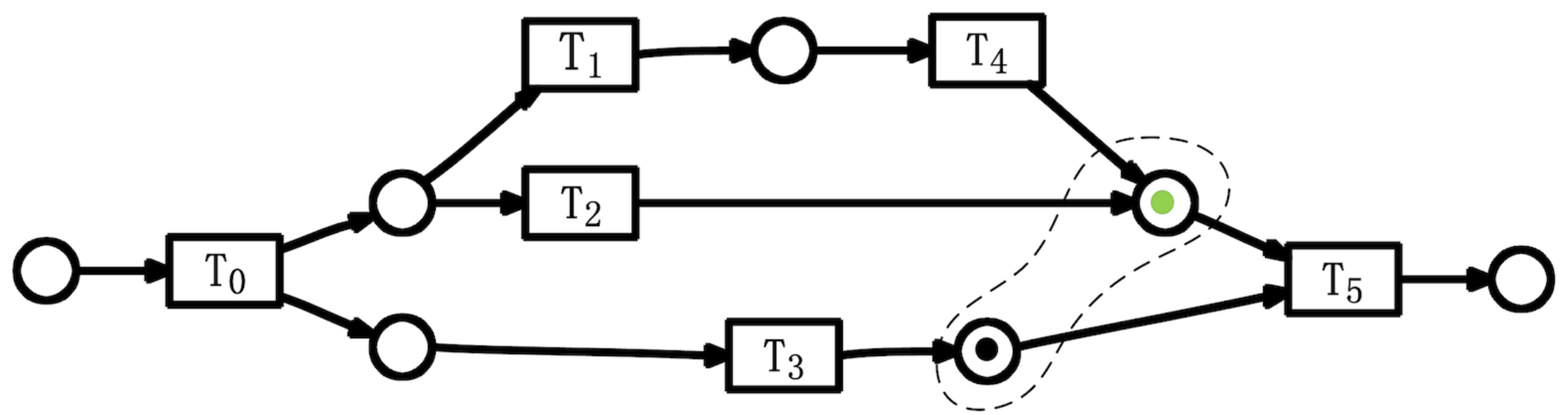}}
  \hspace{0.3in}
  \subfloat[][$T_4\mbox{-}5$ may `re-use' the MEA structures between $T_2\mbox{-}4$ and $T_5\mbox{-}6$]{
    \label{fig:subfigjump:b} 
    \includegraphics[width=3.2in]{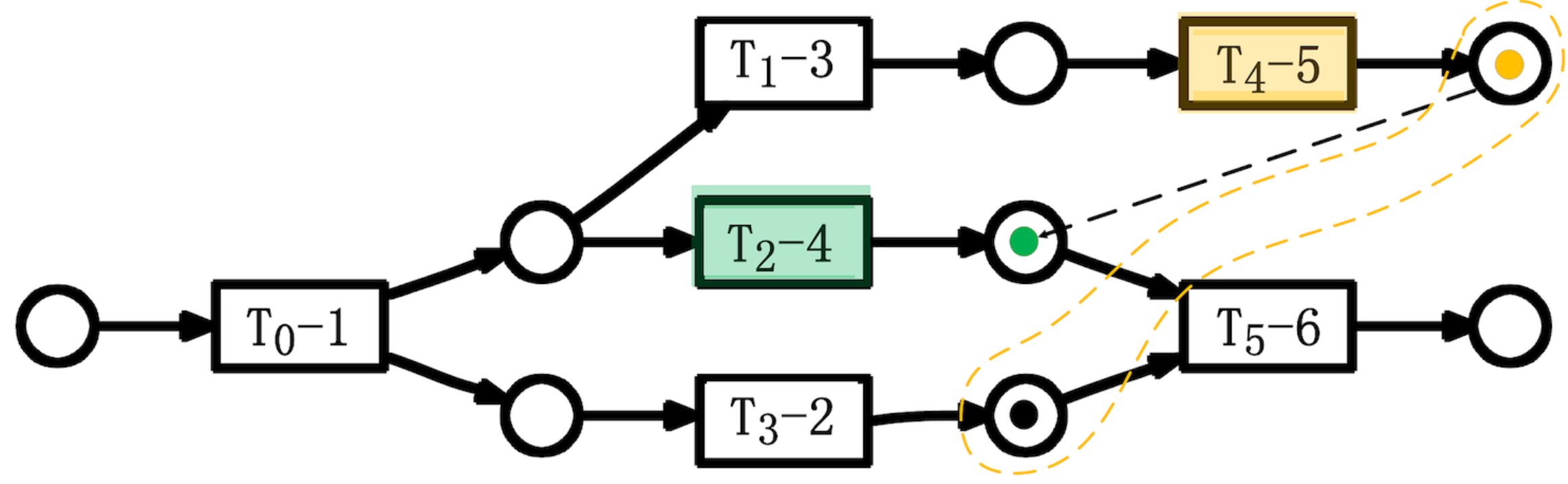}}\\
    \label{fig:subfigjump} 
  \caption{Applying Max-Event Adjacent Structure for TAR Derivation}
\end{figure}

However, MEA based TAR checking requires that the whole MEA structure is preserved in the CFP. In Fig. \ref{fig:subfigjump:b}, as $T_4\mbox{-}5$ is a cut-off event, Thm. \ref{prop:direct} cannot be directly applied. To deal with such cases, we provide the following result:

\begin{theorem}\label{prop:cut}
Let $\pi$ be the CFP of a bounded net system. Let $e_1$ be a cut-off event and $e_1'$ its corresponding event of $\pi$ and $h(e_1\bullet)=h(e_1'\bullet)$. Then $e_1'\triangleright e_2\Rightarrow h(e_1)<_{tar}h(e_2)$. 
\end{theorem}

\begin{proof}
By $e_1'\triangleright e_2$, we know that $e_1'\in Max(\bullet(\bullet e_2))$, and therefore $e_1'<e_2$ and $[e_1']\subseteq[e_2]\backslash\{e_2\}$. By Lemma~\ref{lemm:third}, we can construct a sequence of events $g_1,g_2,...,g_k$ such that $[e_1']=[e_2]\backslash\{e_2\}\backslash\{g_k\}\backslash\{g_{k-1}\}\backslash...\backslash\{g_1\}$. Let $h(g_1),...,h(g_k)$ be denoted as $t_1,...,t_k$. By recursively applying Lemma~\ref{lemm:third} we know that $Mark([e_1'])\xrightarrow{t_1}M_1\xrightarrow{t_2}...\xrightarrow{t_k}Mark([e_2]\backslash\{e_2\})$. From $Mark([e_1])=Mark([e_1'])$ we have $Mark([e_1])\xrightarrow{t_1}M_1\xrightarrow{t_2}...\xrightarrow{t_k}Mark([e_2]\backslash\{e_2\})$ Since $e_1'\triangleright e_2$, only tokens in $Cut([e_1'])\backslash e_1'\bullet$ are required to enable the previous firing sequence. Because of this and $h(e_1'\bullet)=h(e_1\bullet)$, it can be inferred that $Mark([e_1]\backslash\{e_1\})\xrightarrow{t_1}M'_1\xrightarrow{t_2}...\xrightarrow{t_k}M'_k\xrightarrow{h(e_1)}Mark([e_2]\backslash\{e_2\})$. And therefore $h(e_1)<_{tar}h(e_2)$.
\end{proof}

Given the definition of MEA structures, in cases when $t_1\bullet\cap \bullet t_2\neq\emptyset$, we may quickly confirm $t_1<_{tar}t_2$ (if it holds) if we find some event pairs related to $t_1,t_2$ satisfy the conditions of Thm. \ref{prop:direct} and \ref{prop:cut}, which takes much less time to confirm than the coverability check in Algorithm~\ref{alg:basic}.


Using the above results, we can confirm the majority of TARs quickly, except for cases when TAR of a pair of consecutive transitions ($t_1\bullet\cap \bullet t_2\neq\emptyset$) cannot be confirmed using Thm. \ref{prop:direct} and \ref{prop:cut}. In Fig. \ref{fig:weird}, we show such a special example, where additional checking using Algorithm~\ref{alg:basic} is necessary to confirm TAR.

\begin{figure*}
  \makebox[1.1\textwidth][l]{
  \subfloat[][$T_7<_{tar}T_9$]{
    \label{fig:weird:a} 
    \includegraphics[width=0.40\textwidth]{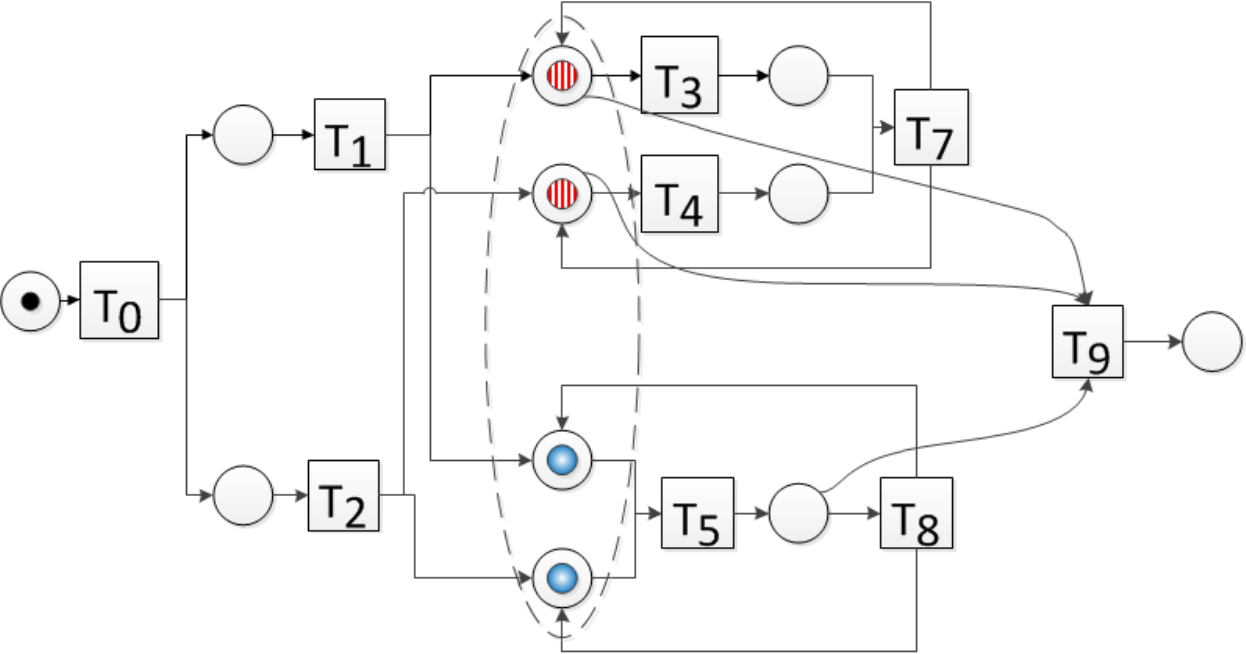}}
  \subfloat[][Thm. \ref{prop:cut} cannot be applied between $T_7\mbox{-}7$ and $T_8\mbox{-}8$]{
    \label{fig:weird:b} 
    \includegraphics[width=0.6\textwidth]{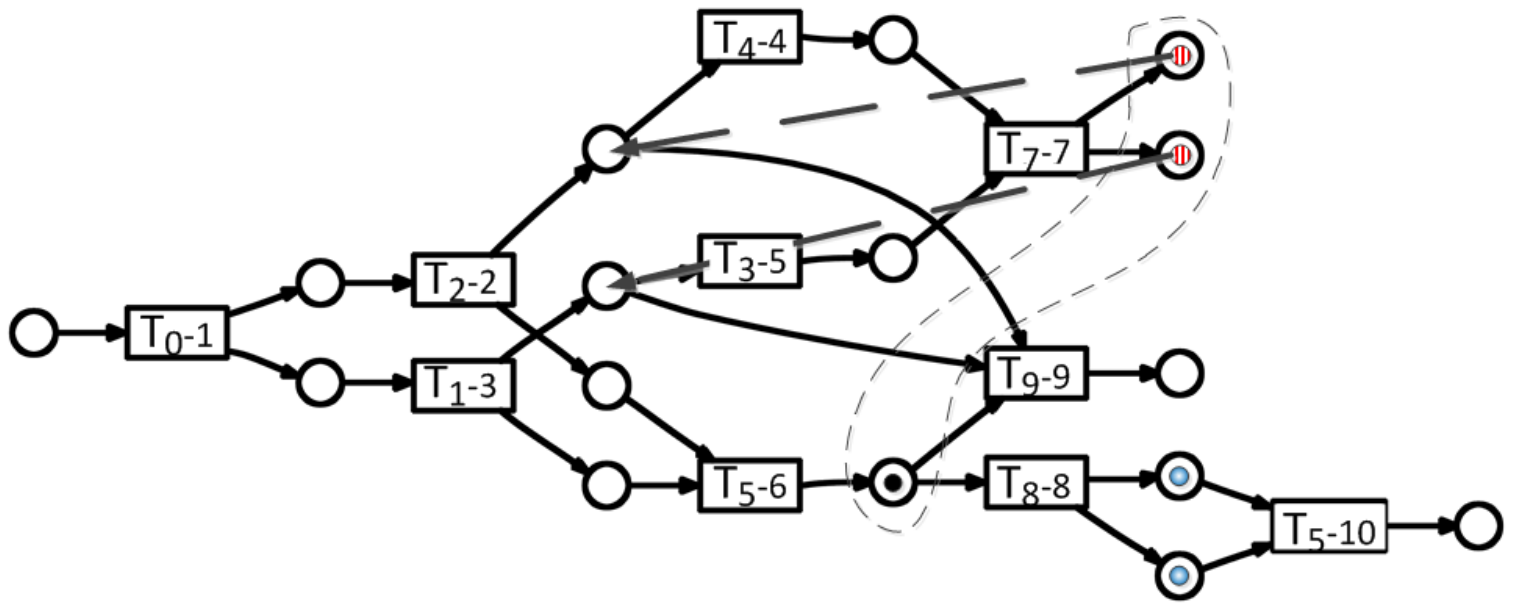}}
  \caption{A challenging cut-off event example that MEA rule does not cover}
  \label{fig:weird} 
  }
\end{figure*}

In Fig. \ref{fig:weird}, the TAR between the consecutive transitions $T_7$ and $T_9$ cannot be confirmed directly using MEA, because $T_7\mbox{-}7$ is a \emph{cut-off} event. Neither can we discover this TAR by reusing the MEA related to $T_7\mbox{-}7$'s corresponding event $T_8\mbox{-}8$ based on Thm. \ref{prop:cut}, because Thm. \ref{prop:cut} requires the corresponding places of the post conditions of the two events are identical, but in this case $h(T_7\mbox{-}7\bullet)\neq h(T_8\mbox{-}8\bullet)$.

Therefore for such pairs of transitions, we are still faced by the challenge that MEA-based derivation rules cannot cover all TAR cases. For these cases, we need to utilize Algorithm~\ref{alg:basic} for general TAR checking.

What we have learned about MEA in the original TAR computation can also be easily extended to pTAR computation. By extending MEA to `neglect' the silent events (events corresponding to silent transitions), we may identify pTAR from CFP using causal relations only.

\begin{definition}[Silent Adjacency]\label{def:sadj}
For two nodes $n_1,n_2\in C\cup E$ in an unfolding $\pi=(O,h), O=(C,E,G)$ such that $n_1<n_2$, we say that $n_1$ and $n_2$ are in silent adjacency, denoted as $n_1\triangleright_s n_2$ if and only if $\forall e\in E:n_1<e<n_2\Rightarrow h(e)$ is a silent transition.
\end{definition}

\begin{remark}\label{rmk:1}
It follows that $n_1\triangleright_s n_2$ iff $n_1<n_2$ and $\forall n_p\in\bullet n_2:n_1<n_p\Rightarrow n_1\triangleright_s n_p$.
\end{remark}

\begin{theorem}\label{thm:silentadj}
For a pair of transitions $t_1,t_2$ in a net system $S$, $t_1<_{ptar}t_2$ if and only if there exist in unfolding a pair of events $e_1,e_2:h(e_1)=t_1, h(e_2)=t_2$ and $e_1\triangleright_{s} e_2$.
\end{theorem}

\begin{proof}
($\Rightarrow$)For any given pair of transitions $t_1,t_2$ in pTAR relation, according to Def. \ref{def:STAR} there exists a marking $M_s$ from which we can consecutively fire a sequence of transitions $\sigma=\langle t_1,t_{i_1},t_{i_2},...,t_{i_n},t_2\rangle$, where $t_{i_1},t_{i_2},...,t_{i_n}$ are silent tasks, i.e. $M_s\xrightarrow{t_1}M_1\xrightarrow{t_{i_1}}M_2\xrightarrow{t_{i_2}}$...$\xrightarrow{t_{i_n}}M_{n+1}\xrightarrow{t_2}M_{n+2}$. In the full unfolding $Unf$ of the net system $S$, there is a configuration $\mathcal{C}$ in $\beta$ for which it holds that $Cut(\mathcal{C})=M_s$. Consequently, a sequence of events $E=\langle e_1, e_{i_1},e_{i_2},...,e_{i_n},e_2\rangle$: $h(e_{i_k})=t_{i_k}(1\leq k\leq n), h(e_1)=t_1, h(e_2)=t_2$ can be added to $\mathcal{C}$. For each event in $\{e_{i_k}\}$ that occurs after $e_1$, because $\mathcal{C}\cup E$ is a configuration, we have $\neg(e_1 \# e_{i_k})$, and therefore it holds either $e_1$ $co$ $e_{i_k}$ or $e_1<e_{i_k}$ ($e_{i_k}<e_1$ does not hold because $e_1$ can be added before $e_{i_k}$). Similarly, we know that $e_1<e_2$, because $e_1$ co $e_2$ must not hold (otherwise, we would have $t_1<_{tar}t_2$, which contradicts with $t_1<_{ptar}t_2$). Consider any event $e:e_1<e<e_2$. Since $e_1<e$, we know that $e\not\in\mathcal{C}$, and consequently $e\in E$. Therefore, $e$ must belong to $\{e_{i_k}\}$, and $h(e)$ is a silent transition. By Def. \ref{def:sadj}, we have $e_1\triangleright_{s} e_2$.

($\Leftarrow$) Let $E_1=\{e_{i_k}\in[e_2]|e_1\leq e_{i_k}< e_2\}$=$\{e_{i_1},e_{i_2},...,e_{i_n},n=|E_1|\}$, where $e_{i_p}<e_{i_q}\Rightarrow p<q$. As $e_1\triangleright_{s}e_2$ it holds that $e_{i_1}=e_1$ and $\{h(e_{i_2}), h(e_{i_3}),...,h(e_{i_n})\}$ are all silent transitions. By Lemma \ref{lemm:first} and Lemma~\ref{lemm:third}, $[e_2]\backslash\{e_2\}$ is a configuration, and $h(e_2)$ is enabled at $M_{n+1}=Mark([e_2]\backslash\{e_2\})$.

Consider $e_{i_n}$. It holds that $\forall e\in [e_2]\backslash\{e_2\}, e_{i_n}\not<e$, otherwise it contradicts with $e_{i_p}<e_{i_q}\Rightarrow p<q$. Therefore $e_{i_n}\in Max([e_2]\backslash\{e_2\})$. Again, by Lemma~\ref{lemm:first} and Lemma~\ref{lemm:third}, $[e_2]\backslash\{e_2, e_{i_n}\}$ is a configuration, and $h(e_{i_n})$ is enabled at $M_n=Mark([e_2]\backslash\{e_2, e_{i_n}\}$ and $M_n\xrightarrow{h(e_{i_n})}M_{n+1}$. By recursively applying the above process, we have $e_{i_k}\in Max([e_2]\backslash\{e_2,e_{i_n},...,e_{i_{k+1}}\}) (k\geq 1)$. In other words, from $M_1=Mark[e_1]$, there exists a firing sequence $M_1\xrightarrow{h(e_{i_1})}M_2\xrightarrow{h(e_{i_2})}...\xrightarrow{h(e_{i_n})}M_{n+1}$, and $h(e_2)$ is enabled at $M_{n+1}$. As $e_{i_1}=e_1$ and $\{h(e_{i_2}), h(e_{i_3}),...,h(e_{i_n})\}$ are all silent transitions, we have $h(e_1)=t_1<_{ptar}t_2=h(e_2)$.
\end{proof}

In order to recover the truncated silent adjacency structures, we propose to extend the original CFP after cut-off events, until \emph{at least one} silent adjacency structure can be identified for each pair of transitions in pTAR.


\begin{figure}
\includegraphics[width=3.0in]{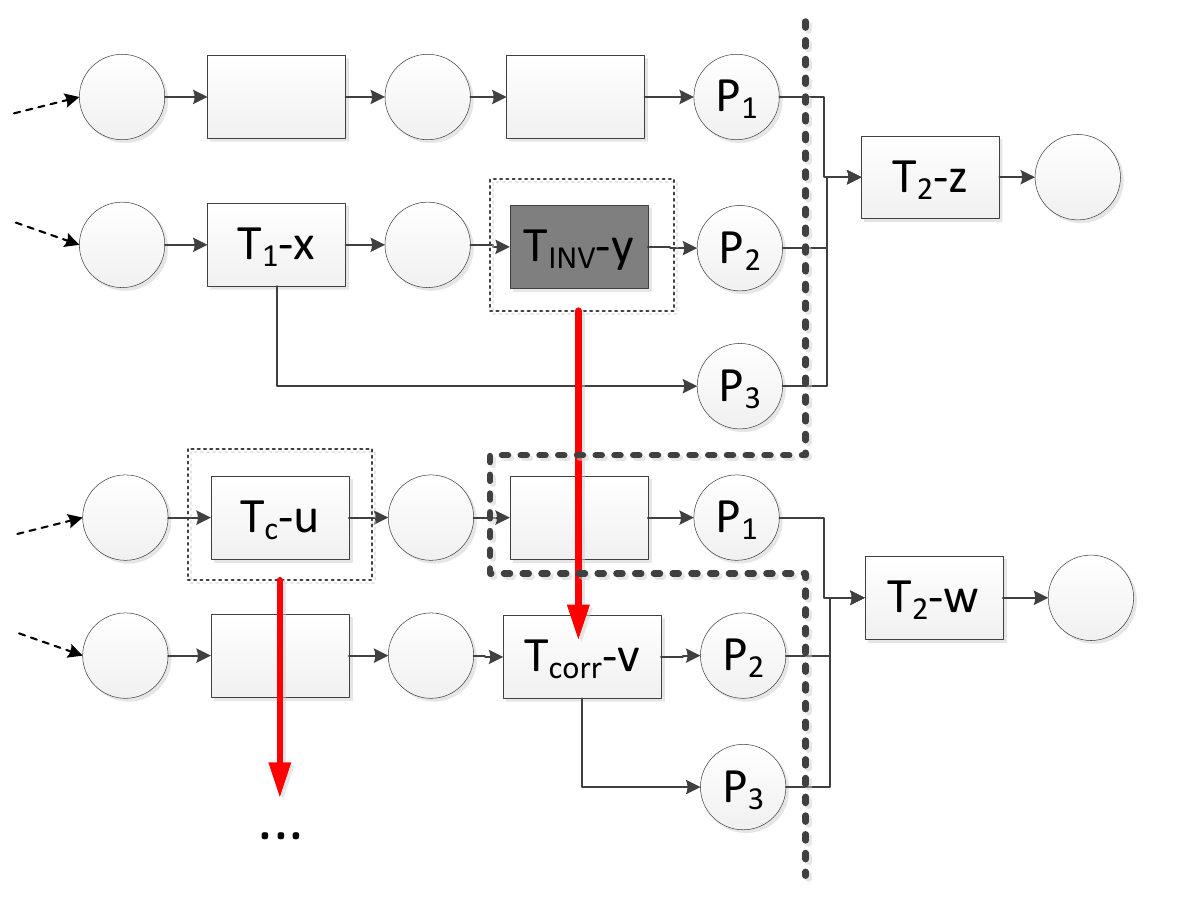}
\caption{The handling of cut-off event `chains' is a major challenge}\label{fig:hard}
\end{figure}

However, because the events in silent adjacency structures may be truncated by cut-off events, the silent adjacency structures may not be completely available in CFP. Fig.~\ref{fig:hard} shows a possible CFP, where a cut-off event $T_{INV}\mbox{-}y$ is formed by a silent transition in a loop structure. It causes the $\triangleright_s$ structure between $T_1\mbox{-}x$ and $T_2\mbox{-}z$ to be truncated.

The idea of looking behind the corresponding events does not work here, because it is not guaranteed that there is an isomorphic silent adjacency structure after $T_{corr}\mbox{-}v$. For example, there might be other cut-off events (in this example, it is $T_c\mbox{-}u$) that are in $co$ with $T_{corr}\mbox{-}v$ and are precedent to the silent adjacency structure after $T_{corr}\mbox{-}v$. In this case, $T_c\mbox{-}u$ truncates $T_2\mbox{-}w$ after $T_{corr}\mbox{-}v$, so now we are forced to seek an isomorphic silent adjacency structure after $T_c\mbox{-}u$'s corresponding event that may confirm $T_1<_{ptar}T_2$. If such patterns keep repeating, it will make efficient computation impossible. One direct solution is to use coverability checking as in Algorithm~\ref{alg:basic} to deal with cases involving cut-off events, as Algorithm~\ref{alg:final} does.
\begin{algorithm}
  \caption{Accelerated TAR Derivation (IMPROVED)}\label{alg:final}
  \begin{algorithmic}[1]
    \Function {IMPROVED}{$S,\pi$}
        \State $TAR\leftarrow\emptyset$
%
%
        \State\textbf{for}{ \textbf{each} $e_1\in \pi$} \textbf{do}
            \State\indent\textbf{for}{ \textbf{each} $e_2:(e_1\triangleright e_2)\vee (e_1$ \emph{co} $e_2)\vee (e_1\triangleright_s e_2)$} \textbf{do}
            \State\indent\indent$TAR\leftarrow TAR\cup \{\langle h(e_1),h(e_2)\rangle\}$
        \State\textbf{for}{ \textbf{each} cut-off event $e$ and its corr. event $e'$} \textbf{do}
            \State \indent\textbf{if} $h(e\bullet)=h(e'\bullet$) \textbf{then}
            \State\indent\indent\textbf{for} \textbf{each} $e_2:e'\triangleright e_2$ \textbf{do} $TAR\leftarrow TAR\cup \{\langle h(e),h(e_2)\rangle\}$
        \State\textbf{for}{ \textbf{each} $t_1,t_2\in S$} \textbf{do}
           \State /*we can further restrict to $t_1,t_2:t_1\bullet\cap \bullet t_2\neq\emptyset$ if we do not compute pTAR*/
                \State\indent\textbf{if} $\langle h(e_1),h(e_2)\rangle\not\in TAR$ \textbf{then}
                \State\indent\indent Reduce $t_1<_{tar}t_2$ or $t_1<_{ptar}t_2$ to a marking $M$'s coverability
                \State\indent\indent\textbf{if} $M$ is coverable \textbf{then} $TAR\leftarrow TAR\cup \{\langle t_1,t_2\rangle\}$
       \State \textbf{return} $TAR$
    \EndFunction
  \end{algorithmic}
\end{algorithm}

From Line 3$\sim$8), Algorithm~\ref{alg:final} attempts to early-confirm TAR/pTAR using the accelerative rules. After that, coverability based approach is applied only to those transition pairs that have not been early-confirmed.

\section{Computing TAR/pTAR with Continued Unfolding}\label{sec:ConUnf}

Apart from coverability checking, another strategy could be extending the CFP structure after cut-off events. Through the following analysis, we show how to extend just enough transitions for pTAR computation and avoid unnecessary extensions.

\subsection{Continued Unfolding After Cut-off Events}

\begin{definition}[Local Silent Extension]\label{def:lse} Let $\beta$ be the full unfolding of a net system $S$. The \textit{local silent extensions} of $\beta$ with respect to an event $e_1$, denoted as $LSE(\beta, e_1)$, is defined as a set of events $\{e=(t,B)\}$, in which for each event $e=(t,B)$, $B$ is a co-set of conditions in $\beta$ and $t$ is a transitions of $S$, $h(B)=\bullet t$ such that
\begin{itemize}
\item[(1)] $\forall c\in B:(e_1\triangleright_s c)\vee(e_1$ co $c$), and $\exists c\in B: e_1\triangleright_s c$
\item[(2)] $\beta$ contains no event $e$ satisfying $h(e)=t$ and $\bullet e=B$
\end{itemize}
\end{definition}

Intuitively, $LSE(\beta, e_1)$ extends the original branching process $\beta$ by only adding any event that is in silent adjacency relation with $e_1$. By recursively doing this, we get an extended CFP in which we shall find all the silent adjacency structures we need, that is:

\begin{definition}[Silently Continued Branching Process]\label{def:leu}
Let $\pi$ be the CFP of a net system $S$ and $e_1$ an event in $\pi$. The \emph{silently continued branching processes} after $e_1$, denoted as $\beta\in \mathcal{L}(e_1)$ are a set of branching processes of $S$ such that:
\begin{itemize}
\item[-] $\pi\in\mathcal{L}(e)$
\item[-] $\forall \beta':\beta'\in \mathcal{L}(e_1)$ and $e=(t,B)\in LSE(\beta',e_1)$, $\beta=\beta'\cup\{e\}\cup e\bullet\in\mathcal{L}(e)$ .
\end{itemize}
\end{definition}

Intuitively, silently continued branching processes are derived by recursively extending the original CFP behind $e_1$ using $LSE$. We call the maximal silently continued branching process w.r.t. a cut-off event $e_1$ in CFP its \emph{silently continued unfolding}, denoted as $\beta_{e_1}$, which can be constructed with Algorithm~\ref{alg:unfsilent}.

\begin{algorithm}
  \caption{Deriving silently continued unfolding after $e_1$}\label{alg:unfsilent}
  \begin{algorithmic}[1]
    \Function {silently-continued-unfolding}{$S,\pi,e_1$}
        \State $\beta_{e_1}\leftarrow\pi$
        \State $lse\leftarrow LSE(\beta_{e_1}, e_1)$
        \State \textbf{while} $lse\neq\emptyset$ \textbf{do}
        \State\indent take an event $e=(t,B)$ out of $lse$ and add it to $\beta_{e_1}$ with a condition $(c, e)$ for every output place $p$ of $t$.
        \State\indent $lse\leftarrow LSE(\beta_{e_1},e_1)$
        \State \textbf{endwhile}
    \EndFunction
  \end{algorithmic}
\end{algorithm}

Similar to the original unfolding algorithm \cite{CPU02}, Algorithm \ref{alg:unfsilent} does not necessarily terminate. This suggest we need to introduce a new kind of termination condition for silently continued branching processes (i.e. a new kind of cut-off event ). But let us come back to this later. Now, with the denotation of $\beta_{e_1}$ and Algorithm~\ref{alg:unfsilent}, we can prove that for each pair of transitions in pTAR, at least one silent adjacency structure could be found in $\beta_{e_1}$.

\begin{theorem}\label{thm:silentbeta}
Let $\beta$ and $\pi$ be the full unfolding and CFP of a net system $S$ respectively. For every pair $e_1,e_2:e_1\triangleright e_2$ in $\beta$, there exists an event $e_1':h(e_1')=h(e_1)$ in $\pi$, such that in $\beta_{e_1}$, there is an event $e_2':h(e_2')=h(e_2)$ and $e_1'\triangleright e_2'$.
\end{theorem}

\begin{proof}
Let $E=\{e_{i_k}:e_1<e_{i_k}<e_2\}$, $1\leq k\leq n$ and restrict, without loss of generality, that $e_{i_p}<e_{i_q}\Rightarrow p<q$. Since $[e_2]$ is a configuration, by recursively applying Lemma~\ref{lemm:first}, we know that $$\mathcal{C}=[e_2]\backslash\{e_2\}\backslash\{e_{i_k}\}\backslash\{e_{i_{k-1}}\}\backslash...\backslash\{e_{i_1}\}\backslash\{e_1\}$$ is a configuration. Consequently, it is easy to prove that $\mathcal{C}\oplus\{e_1\}$,$\mathcal{C}\oplus\{e_1,e_{i_1}\}$,...,$\mathcal{C}\oplus\{e_1,e_{i_1},...,e_{i_n}\}$ and $\mathcal{C}\oplus\{e_1,e_{i_1},...,e_{i_n},e_2\}=[e_2]$ are all configurations.\\

By Property \ref{prop:first}, we know that in $\pi$ there exists a configuration $\mathcal{C'}$ such that $Mark(\mathcal{C'})=Mark(\mathcal{C})$ and $\mathcal{C'}$ does not contain any cut-off events.

As $Mark(\mathcal{C'})=Mark(\mathcal{C})$, $\Uparrow\mathcal{C'}$ is isomorphic with $\Uparrow\mathcal{C}$. Let $I_1^2$ be the isomorphism between $\mathcal{C}$ and $\mathcal{C'}$. Let $e'_1=I_1^2(e_1)$. Since $\mathcal{C}\oplus \{e_1\}$ is a configuration, $\mathcal{C'}\oplus \{e'_1\}$ is also a configuration. And as $\mathcal{C'}$ does not contain any cut-off event, $\mathcal{C'}\cup\{e_1'\}$ is contained in $\pi$.\\

Let $e'_{i_k}=I_1^2(i_k)$ $(1\leq k\leq n)$, and $e'_2=I_1^2(e_2)$.
Since $\mathcal{C}\oplus\{e_1,e_{i_1},...,e_{i_k},e_2\}$ is a configuration, $\mathcal{C'}\oplus\{e'_1,e'_{i_1},...,e'_{i_k},e'_2\}$ should also be a configuration in $\beta$ (but not necessarily fully contained in $\pi$).
For Algorithm~\ref{alg:unfsilent} 
, we now prove that $\mathcal{C'}\oplus\{e'_1,e'_{i_1},...,e'_{i_n},e'_2\}$ are in $\beta_{e'_1}$. First, because $\mathcal{C'}\cup\{e'_1\}\subseteq\pi\subseteq\beta_{e'_1}$, $e'_1$ is contained in $\beta_{e'_1}$. For events in $\{e'_{i_k}\}$, we prove by induction:

\textbf{(a)} When $k=1$, because $(\mathcal{C}\cup\{e_1\})\oplus\{e_{i_1}\}$ is a configuration, $(\mathcal{C}\cup\{e'_1\})\oplus\{e'_{i_1}\}$ is also a configuration. By isomorphism, as $e_1<e_{i_1}$, it holds that $e'_1<e'_{i_1}$. Therefore $B=\bullet e'_{i_1}\subseteq Cut(\mathcal{C'}\oplus\{e'_1\})\subseteq\beta_{e'_1}$ and $\exists c\in B: e_1<c$ (condition (1) of Def. \ref{def:sadj}). For any $c$ in $B$, we have $\{ e|e'_1<e<c\}=\emptyset$, so condition (2) of Def. \ref{def:sadj} holds trivially. Therefore, at Line 3 of Algorithm~\ref{alg:unfsilent}, if $\beta_{e'_1}$ does not already contain $e'_{i_1}=(h(e'_{i_1}),B)$, then condition (3) is satisfied as well, and $e'_{i_1}$ will be added to $lse$ and eventually inserted into $\beta_{e'_1}$ at Line 5.

\textbf{(b)} For any $k>1$, assume that $\mathcal{C'}\oplus\{e'_1,e'_{i_1},...,e'_{i_k-1}\}$ can eventually be contained in $\beta_{e'_1}$. If $e'_{i_k}$ is in $\pi$, then the induction hypothesis holds trivially. Next we consider the case of $e'_{i_k}$ not being $\pi$.

Let $B=\bullet e_{i_k}$. By isomorphism, $\mathcal{C'}\oplus\{e'_1,e'_{i_1},...,e'_{i_k-1}\}\oplus\{e'_{i_k}\}$ is also a configuration, so we have $B\subseteq Cut(\mathcal{C'}\oplus\{e'_1,e'_{i_1},...,e'_{i_k-1}\})$. As $e_1<e_{i_k}$, it holds that $e'_1<e'_{i_k}$ and consequently $\exists c\in B: e'_1<c$. Therefore condition (1) of Def. \ref{def:sadj} holds.

Consider $E_M=\{e|e'_1<e<c,c\in B\}$. If $E_M=\emptyset$, then in the same manner of (a) we can conclude that both condition (1), (2) and (3) of Def. \ref{def:sadj} will be satisfied at Line 3 of Algorithm~\ref{alg:unfsilent}, and $e'_{i_k}$ will be added in $\beta_{e'_1}$ if it is not already in $\pi$.Otherwise, if $E_M=\{e|e'_1<e<c,c\in B\}$ is not empty, then it must hold that $E_M\subseteq\{e'_{i_1},...,e'_{i_k-1}\}$. (This is because $c\in Cut(\mathcal{C'}\oplus\{e'_1,e'_{i_1},...,e'_{i_k-1}\})$, and in $\mathcal{C'}\oplus\{e'_1,e'_{i_1},...,e'_{i_k-1}\}$, $\{e'_{i_1},...,e'_{i_k-1}\}$ are the only events added after $\mathcal{C}\oplus\{e'_1\}$.) As $\{h(e'_{i_1}),...,h(e'_{i_k-1})\}$ are all silent transitions, condition (2) for $LSE$ also holds.

If $\mathcal{C'}\oplus\{e'_1,e'_{i_1},...,e'_{i_k-1}\}$ are already contained in $\pi$ in the beginning, then $e'_{i_k}$ can be added to $lse$ at Line 3 and eventually into $\beta_{e'_1}$ at Line 5. Otherwise, by our induction hypothesis, the events in $\{e'_{i_1},...,e'_{i_k-1}\}$ which are not in $\pi$ will be added into $\beta_{e'_1}$ by Line 5. Immediately after the last one of such events are added into $\beta_{e'_1}$ at Line 5, condition (3) for $LSE$ is satisfied, and $e'_{i_k}$ will be added to $lse$ and eventually into $\beta_{e'_1}$ at Line 5.

By \textbf{(a)} and \textbf{(b)}, by induction we know that $\mathcal{C'}\oplus\{e'_1,e'_{i_1},...,e'_{i_k}\}$ is either contained in $\pi$, or can be added to $\beta_{e'_1}$ in the loop of Line 4-8. \\

As we have previously shown, $\mathcal{C'}\oplus\{e'_1,e'_{i_1},...,e'_{i_k}\}\oplus\{e'_2\}$ is also a configuration. By the similar reasoning in \textbf{(b)}, we know that $\mathcal{C'}\oplus\{e'_1,e'_{i_1},...,e'_{i_k}\}\oplus\{e'_2\}$ will eventually be contained in $\beta_{e_1'}$, and that $e'_1<e'_2,\{e|e'_1<e<e'_2\}\subseteq\{e'_{i_1},...,e'_{i_k}\}$. As $\{h(e'_{i_1}),...,h(e'_{i_k})\}$ are all silent transitions, we have $e'_1\triangleright_s e'_2$.
\end{proof}

\subsection{Determine When to Stop a Silently Continued Unfolding}\label{subsec:stop}

Being similar with Esparza's original unfolding algorithm \cite{CPU02}, Algorithm~\ref{alg:unfsilent} 
does not necessarily terminate. To avoid state-explosion and infinite unfolding, we should stop silently continued unfolding when enough information has been unfolded for complete pTAR computation. In order to do this, we introduce a new type of cut-off events in silently continued unfolding as follows:

\begin{definition}\label{def:cut_lse} (Local Cut-off Events in Silently Continued Unfolding) Given two events $e_a,e_b$ in $\beta_{e_1}\backslash\pi$, let $E_a=\{e\in[e_a]:e\not\in \pi\}$, $E_b=\{e\in[e_b]:e\not\in \pi\}$. We say that $e_b$ is in the set of local cut-off events in silently continued unfolding $\beta_{e_1}\backslash\pi$, denoted as cut-off$_{e_1}$, if
\begin{itemize}
\item[(1)] $Min(\bullet E_a)\subseteq Min(\bullet E_b)$,
\item[(2)] $h((E_a\bullet\backslash \bullet E_a)\cup(Min(\bullet E_b)\backslash Min(\bullet E_a)))=h(E_b\bullet\backslash\bullet E_b)$,
\item[(3)] $h(e_a)$ is a silent transition, and
\item[(4)] $|[E_a]|<|[E_b]|$, .
\end{itemize}
\end{definition}

Using Def. \ref{def:cut_lse}, we can truncate a silently continued unfolding to acquire its `complete finite prefix' (denoted as $Fin_{e_1}$ or $\pi_{e_1}$) of $\beta_{e_1}$. Here, `complete' means any pTAR detectable with $\triangleright_s$ in $\beta_{e_1}$ can be discovered from the prefix with $\triangleright_s$ as well, and `finite' simply means the size of the prefix is finite.



Given the definition of our new type of cut-off events, the algorithm for constructing the `complete finite prefix' with respect to the silently continued unfolding of an event $e_1$ can be provided as follows.

\begin{algorithm}
  \caption{Finite Prefix of Continued Unfolding for pTAR Derivation}\label{alg:star}
  \begin{algorithmic}[1]
    \Function {FinitePrefixForPTAR}{$S,\pi,e_1$}
            \State $Fin_{e_1}\leftarrow\pi$
            \State $lse\leftarrow LSE(\pi, e_1)$
            \State \textbf{while} $lse\neq\emptyset$ \textbf{do}
            \State\indent pick an event $e=(t,B)$ from $lse$
            \State\indent $lse\leftarrow lse\backslash\{e\}$
            \State\indent \textbf{if} $[e]\cap$ cut-off$_{e_1}=\emptyset$ \textbf{then}
            \State\indent\indent add $e$ to $Fin_{e_1}$ together with a condition $(c, e)$
            \State\indent\indent \textbf{if} $h(e)$ is not silent \textbf{then} $pTAR\leftarrow pTAR\cup\{\langle h(e_1),h(e)\rangle\}$
            \State\indent\indent \textbf{if} $e$ is a local cut-off event in $Fin_{e_1}\backslash \pi$ \textbf{then}
            \State\indent\indent\indent cut-off$_{e_1}\leftarrow$cut-off$_{e_1}\cup \{e\}$
            \State\indent\indent \textbf{else}
            \State\indent\indent\indent $lse\leftarrow LSE(Fin_{e_1},e_1)$
            \State\indent \textbf{end if}
            \State \textbf{endwhile}
            \State return pTAR

    \EndFunction
  \end{algorithmic}
\end{algorithm}

\begin{theorem}
For a pair of transitions $t_1,t_2$ in a net system $S$, $t_1<_{star}t_2$ if and only if there exist in unfolding (Unf) a pair of events $e_1,e_2:h(e_1)=t_1, h(e_2)=t_2$ and $e_1\triangleright_{s} e_2$.
\end{theorem}
\begin{proof}
($\Rightarrow$)For any given pair of transitions $t_1,t_2$ in STAR relation, according to Def. \ref{def:STAR} there exists a marking $M_s$ from which we can consecutively fire a sequence of transitions $\sigma=\langle t_1,t_{i_1},t_{i_2},...,t_{i_n},t_2\rangle$, where $t_{i_1},t_{i_2},...,t_{i_n}$ are silent tasks, i.e. $M_s\xrightarrow{t_1}M_1\xrightarrow{t_{i_1}}M_2\xrightarrow{t_{i_2}}$...$\xrightarrow{t_{i_n}}M_{n+1}\xrightarrow{t_2}M_{n+2}$. In the full unfolding $Unf$ of the net system $S$, there is a configuration $\mathcal{C}$ in $\beta$ for which it holds that $Cut(\mathcal{C})=M_s$. Consequently, a sequence of events $E=\langle e_1, e_{i_1},e_{i_2},...,e_{i_n},e_2\rangle$: $h(e_{i_k})=t_{i_k}(1\leq k\leq n), h(e_1)=t_1, h(e_2)=t_2$ can be added to $\mathcal{C}$. For each event in $\{e_{i_k}\}$ that occurs after $e_1$, because $\mathcal{C}\cup E$ is a configuration, we have $\neg(e_1 \# e_{i_k})$, and therefore it holds either $e_1$ $co$ $e_{i_k}$ or $e_1<e_{i_k}$ ($e_{i_k}<e_1$ does not hold because $e_1$ can be added before $e_{i_k}$). Similarly, we know that $e_1<e_2$, because $e_1$ co $e_2$ must not hold (otherwise, we would have $t_1<_{tar}t_2$, which contradicts with $t_1<_{star}t_2$). Consider any event $e:e_1<e<e_2$. Since $e_1<e$, we know that $e\not\in\mathcal{C}$, and consequently $e\in E$. Therefore, $e$ must belong to $\{e_{i_k}\}$, and $h(e)$ is a silent transition. By Def. \ref{def:sadj}, we have $e_1\triangleright_{s} e_2$.

($\Leftarrow$) Let $E_1=\{e_{i_k}\in[e_2]|e_1\leq e_{i_k}< e_2\}$=$\{e_{i_1},e_{i_2},...,e_{i_n},n=|E_1|\}$, where $e_{i_p}<e_{i_q}\Rightarrow p<q$. As $e_1\triangleright_{s}e_2$ it holds that $e_{i_1}=e_1$ and $\{h(e_{i_2}), h(e_{i_3}),...,h(e_{i_n})\}$ are all silent transitions. By Lemma X and Lemma Z, $[e_2]\backslash\{e_2\}$ is a configuration, and $h(e_2)$ is enabled at $M_{n+1}=Mark([e_2]\backslash\{e_2\})$.

Consider $e_{i_n}$. It holds that $\forall e\in [e_2]\backslash\{e_2\}, e_{i_n}\not<e$, otherwise it contradicts with $e_{i_p}<e_{i_q}\Rightarrow p<q$. Therefore $e_{i_n}\in Max([e_2]\backslash\{e_2\})$. Again, by Lemma X and Lemma Z, $[e_2]\backslash\{e_2, e_{i_n}\}$ is a configuration, and $h(e_{i_n})$ is enabled at $M_n=Mark([e_2]\backslash\{e_2, e_{i_n}\}$ and $M_n\xrightarrow{h(e_{i_n})}M_{n+1}$. By recursively applying the above process, we have $e_{i_k}\in Max([e_2]\backslash\{e_2,e_{i_n},...,e_{i_{k+1}}\}) (k\geq 1)$. In other words, from $M_1=Mark[e_1]$, there exists a firing sequence $M_1\xrightarrow{h(e_{i_1})}M_2\xrightarrow{h(e_{i_2})}...\xrightarrow{h(e_{i_n})}M_{n+1}$, and $h(e_2)$ is enabled at $M_{n+1}$. As $e_{i_1}=e_1$ and $\{h(e_{i_2}), h(e_{i_3}),...,h(e_{i_n})\}$ are all silent transitions, we have $h(e_1)=t_1<_{star}t_2=h(e_2)$.
\end{proof}

We prove the finiteness and completeness of Algorithm~\ref{alg:star} with the following two results.

\begin{theorem}\label{thm:finfinite}
$Fin_{e_1}$ is finite.
\end{theorem}

\begin{proof}
Since $\pi$ is already finite. We only focus on the part of $Fin_{e_1}\backslash\pi$. Let $Q=\{B\}$ be the set of co-sets in $Fin_{e_1}$. As each $B\in Q$ corresponds to part of cut, each $h(B)$ corresponds to a sub-set of a reachable marking in $S$. Since $S$ has a finite number of reachable markings, the number of distinct $h(B)$ are also finite. Let $n$ be the number of different $h(B)$.

Given an event $e$ of $Fin_{e_1}\backslash\pi$, define the \emph{depth} of e as the length of a longest chain of events $g_1<g_2<...<e$ such that $g_1,g_2,...,e$ are all in $Fin_{e_1}\backslash\pi$; the depth of $e$ is denoted by $d(e)$.

Moreover, the following results hold.

\textbf{(1)} For every event $e$ of $Fin_{e_1}\backslash\pi$, $d(e)\leq n+1$, where $n$ is the number of different $h(B)$.

For every chain of events $g_1<g_2<...<g_{n+1}$, let $E_i=[g_i]\backslash\pi$, and we know that $i<j\Rightarrow Min(\bullet E_i)\subseteq Min(\bullet E_{j})$. We denote $B_i=(E_i\bullet\backslash\bullet E_i)\cup(Min(\bullet E_{n+1})\backslash Min(\bullet E_{i}))$.

Obviously every $B_i$ is a co-set. So there exist $i,j:i<j$ such that $h(B_i)=h(B_j)$, i.e. $h((E_i\bullet\backslash\bullet E_i)\cup(Min(\bullet E_{n+1})\backslash Min(\bullet E_{i})))=h((E_j\bullet\backslash\bullet E_j)\cup(Min(\bullet E_{n+1})\backslash Min(\bullet E_{j})))$. By removing $Min(\bullet E_{n+1})\backslash Min(\bullet E_j)$ from both sides, we get $h((E_i\bullet\backslash\bullet E_i)\cup(Min(\bullet E_j)\backslash Min(\bullet E_{i})))=h(E_j\bullet\backslash\bullet E_j)$. Since $g_i<g_j$, we have $E_i\subset E_j$ and therefore $|E_i|<|E_j|$, so if $h(g_i)$ is a silent transition, then $g_j$ should be recognized as a local cut-off event in $\beta_{e_1}$. Should Algorithm~\ref{alg:star} generate $g_j$, then it has generated $g_i$ before and $g_j$ is recognized as a local cut-off event of $Fin_{g_1}$, too. If $h(g_i)$ is not a silent transition, then Algorithm~\ref{alg:star} will not generate $g_{i+1},...,g_{n+1}$ if it has generated $g_i$ before, because as $e_1<g_i$, $e_1\triangleright_s g_{i+1}$ cannot hold.

\textbf{(2)} For every event $e$ of $Fin_{e_1}$, the sets $\bullet e$ and $e\bullet$ are finite. And for every $k\geq 0$, $Fin_{e_1}$ contains only finitely many events $e$ such that $d(e)\leq k$

It follows from \textbf{(1)} and \textbf{(2)} that $Fin_{e_1}$ contains only finitely many events, and by (2) it contains a finite number of conditions.
\end{proof}

\begin{theorem}\label{thm:fincomplete}
$Fin_{e_1}$ is complete in the sense that if $e_1\triangleright_s e_2$ in the full unfolding $\beta$, then $e_2\in Fin_{e_1}$.
\end{theorem}

\begin{proof}
We denote the original net system, its full unfolding and complete finite prefix as $S,\beta$及$\pi$ respectively. For a pair of events $e_a,e_b\in\beta: e_a\triangleright_s e_b$, by Thm. \ref{thm:silentbeta} we have $\exists e_1\in\pi, e_2\in\beta_{e_1}: h(e_1)=h(e_a), h(e_2)=h(e_b), e_1\triangleright_s e_2$, in which $\beta_{e_1}$ is the silently continued unfolding after $e_1$. If $[e_2]$ is not already contained by $\pi_{e_1}$, it must contain some cut-off event $e_c\subseteq$cut-off$_{e_1}$ in $\beta_{e_1}$. For such a cut-off event $e_c$, let $E_c=[e_c]\backslash\pi$, $E_1=[e_2]\backslash\pi\backslash E_c$. Furthermore let $e'_c$ be the corresponding event of $e_c$, $E'_c=\{[e'_c]\backslash\pi\}$. Obviously, there is an isomorphic mapping $I^{2}_1$ between $\Uparrow[e_c]$ and $\Uparrow[e_c']$ such that $I^{2}_1(E_2\backslash [e_c])\subseteq\beta_{e_1}$, in which there is an event $e_2'=I^{2}_1(e_2)$ and $e_1\triangleright_s e_2',h(e_2')=h(e_2)$.

According to the condition (4) of Def. \ref{def:cut_lse}, we have $|E_c'|=|[e'_c]\backslash\pi|<|[e_c]\backslash\pi|=|E_c|$. On the other hand, for $E_2=[e_2]\backslash\pi=E_c\cup E_1$ and $E'_2=[e'_2]\backslash\pi\subseteq E'_c\cup I_1^2(E_1)$, we have $|E'_2|<|E_2|$ because $|E'_c|<|E_c|$. If $[e'_2]$ is still not fully contained in $\pi_{e_1}$, then we may repeat the above procedure and find $e''_2:h(e''_2)=h(e'_2)=h(e_2), e_1\triangleright_s e''_2$ in $\beta_{e_1}$(noticing condition (3) of Def. \ref{def:cut_lse}), and for $E''_2=[e''_2]\backslash\pi$ we still have $|E''_2|<|E'_2|<|E_2|$. Since $|E|\geq 0$, this procedure can not be repeated infinitely and must terminate at some event $e_t$ such that $e_1\triangleright_s e_t, h(e_t)=h(e_2), E_t=[e_t]\backslash\pi$ is already contained in $\pi_{e_1}$. Therefore in Algorithm~\ref{alg:star}, $\pi_{e_1}$ is complete, i.e. for some event $e_2$ in $\beta_{e_1}$, if $e_1\triangleright_s e_2$, it must hold that $\exists e_t\in \pi_{e_1}:e_1\triangleright_s e_t,h(e_t)=h(e_2)$。
\end{proof}

\subsection{Prevent Unnecessary Extensions}

Based on Thm.~\ref{thm:fincomplete}, by building finite prefixes of the silently continued unfoldings after each event in the original CFP, we can derive all pTARs using only silent adjacency and $co$ relations. We use the following definition to denote the `sum' of all silently continued CFPs.

\begin{definition}[Merge of Silently Extended CFPs] Let $E=\{e_i\}$ be the set of events of a complete finite prefix. Then the merged CFP of the locally extended CFPs after these unfinished events is denoted as: $$\mu(\pi)=\displaystyle \bigcup_{e_i\in E}Fin_{e_i}.$$
\end{definition}

By the above results, it is easy to see that every pair of transitions in pTAR correspond to at least one pair of events in $\mu(\pi)$ that are in silent adjacency. And therefore the pTAR computation algorithm based on silently continued unfolding can be given as Algorithm~\ref{alg:continued}.

\begin{algorithm}
  \caption{Deriving all pTARs based on silently continued unfolding}\label{alg:continued}
  \begin{algorithmic}[1]
  \label{alg:derivestar}
    \Function {UnfoldLocal}{$S,\pi$}
        \State $pTAR\leftarrow\emptyset$
        \State $Extended\leftarrow\emptyset$
        \State \textbf{for each} $e_1,e_2\in\pi: e_1\triangleright_s e_2$ \textbf{do}
        \State \indent $pTAR\leftarrow pTAR\cup\langle h(e_1),h(e_2)\rangle$
        \State $C_{init}\leftarrow$ cut-off events of $\pi$
        \State \textbf{for each} $e_c\in C_{init}$ \textbf{do}
        \State \indent \textbf{for each} $e_1\in\pi:e_1\triangleright_s e_c$ \textbf{do}
        \State \indent \indent \textbf{if} $e_1\not\in Extended$ \textbf{then}
        \State \indent \indent\indent $Extended\leftarrow Extended\cup\{e_1\}$
        \State \indent \indent\indent $Fin_{e_1}=\Call{FinitePrefixForPTAR}{S,\pi,e_1}$
        \State \indent \indent\indent \textbf{for each} $e_b\in Fin_{e_1}\backslash\pi:e_a\triangleright_s e_b$ \textbf{do} 
        \State \indent \indent\indent \indent $pTAR\leftarrow pTAR\cup\langle h(e_a),h(e_b)\rangle$
        \State \indent \indent\indent $\pi\leftarrow\pi\cup Fin_{e_1}$
        \State \indent \indent \textbf{end if}
        \State \textbf{return} $pTAR$

    \EndFunction
  \end{algorithmic}
\end{algorithm}

However, the above algorithm is only a baseline result, and can be further improved. In fact, we notice that the silent continued unfoldings might overlap with existing structures inside the original CFP. These overlapped parts do not offer new information for pTAR derivation, and incur serious amount of computational costs (see the following figure).

\begin{figure*}
    \label{fig:lseex1} 
    \includegraphics[width=6.6in]{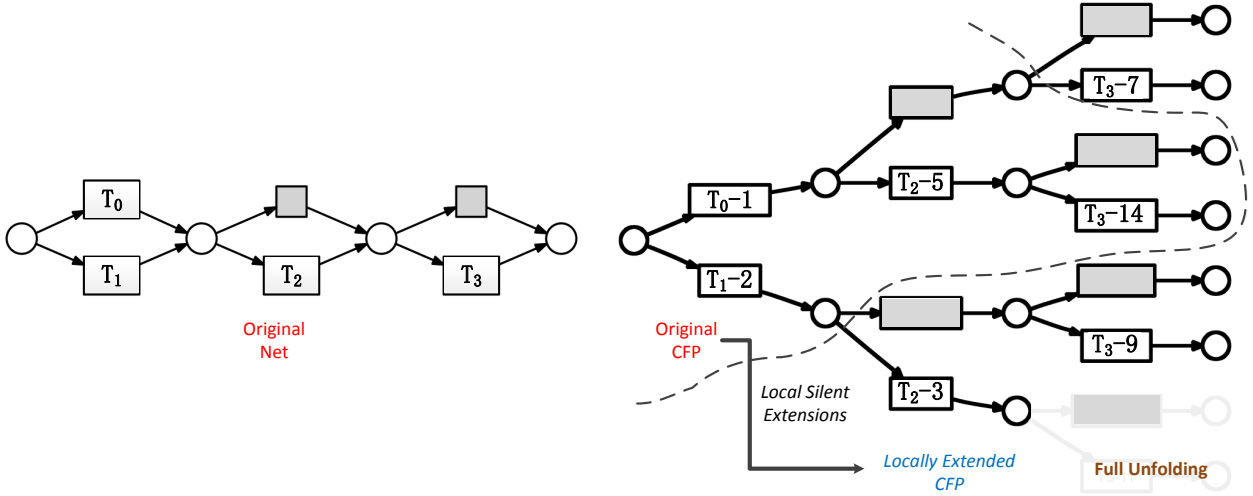}
    \caption{Merged continued CFPs can be as large as the reachability graph.}
\end{figure*}

In such cases, we may re-use existing structures in CFP as an isomorphism of the behavior after cut-off events. Next, we show that this can be realized by redirecting the post-arcs of some cut-off events. The basic purpose is to `link' the cut-off events back to the existing nodes in the CFP. Actually, we may do this when the cut-off events have identical `context' with its corresponding event, as the following definition shows.

\begin{definition}[Contextual Isomorphism and Redirected CFP]\label{def:context} In an unfolding prefix $\pi$, let $e$ be a cut-off event and $e'$ its corresponding event such that $e'$ is not a cut-off event. Let $\gamma: e\bullet\rightarrow Cut([e'])$ denote an injective mapping from $e\bullet$ to the conditions in $Cut([e'])$. If there exists a mapping $\gamma$ such that
\begin{itemize}
\item[(1)] $h(c)=h(\gamma(c)),c\in e\bullet$, and
\item[(2)] $\forall c\in e\bullet, c'\not\in (e\bullet\cup \gamma(e\bullet)):c' $ co $c\Leftrightarrow c'$ co $\gamma(c)$.

\end{itemize}
we say that $e$ is contextually isomorphic with $e'$ under the mapping $\gamma$ (denoted as $e \stackrel{\gamma}{\sim} e'$).
\end{definition}

Note that the above properties are not implied by the definition of cut-off events, which only guarantees $Cut([e])$ and $Cut([e'])$ are mapped to identical set of places. However, we notice they can be applied for most of the cut-off events in the practical models.

When a cut-off event $e$ is contextually isomorphic with some $e'$, we can `link' $e$ towards the conditions in $Cut([e'])$ according to the $\gamma$ mapping so as to reuse the existing structure after $Cut([e'])$. For every condition $c$ in $e\bullet$, we redirect the arc following $e$ from their original target condition $c$ to $\gamma(c)$, and we call these arcs \emph{redirected arcs}. After that, we remove the original conditions in $e\bullet$ from the CFP. Note that the redirected arcs introduced should not be taken as the same as the original arcs in the unfolding, otherwise, they will disrupt the restriction of branching process that $|\bullet c|\leq 1$, and it may also disrupt the acyclic property of the original causal relation `$<$'. Therefore the redirected arcs are used separately for pTAR derivation only. We use $pre(c)$ to denote the set of starting events whose post-arcs are redirected to a condition $c$, which is defined as:
$$
    pre(c)=\{e\in\pi|\exists c'\in e\bullet, e'\in\pi: e\stackrel{\gamma}{\sim}e'\wedge\gamma(c')=c\}
$$

Intuitively, we could call the events in $pre(c)$, whose post-arcs are redirected, as `\emph{soft}' cut-off events. We denote the redirected CFP of the original CFP $\pi$ as $(\pi\backslash E_c\bullet,pre)$, where $E_c$ is the set of `\emph{soft}' cut-off events and $pre$ indicates the redirected arc relations.

In the following example, we could see that the truncated behaviors after cut-off event $T_2\mbox{-}4$ and the two silent events are consecutively continued after the conditions which redirected arcs point to. We call such events $soft$ cut-off events because we can find their succeeding behaviors by contextual isomorphism without performing further local extensions.

\begin{figure*}
  \centering
  \subfloat[][A net with multiple cut-off events]{
    \label{fig:lseex2:a} 
    \includegraphics[width=4.4in]{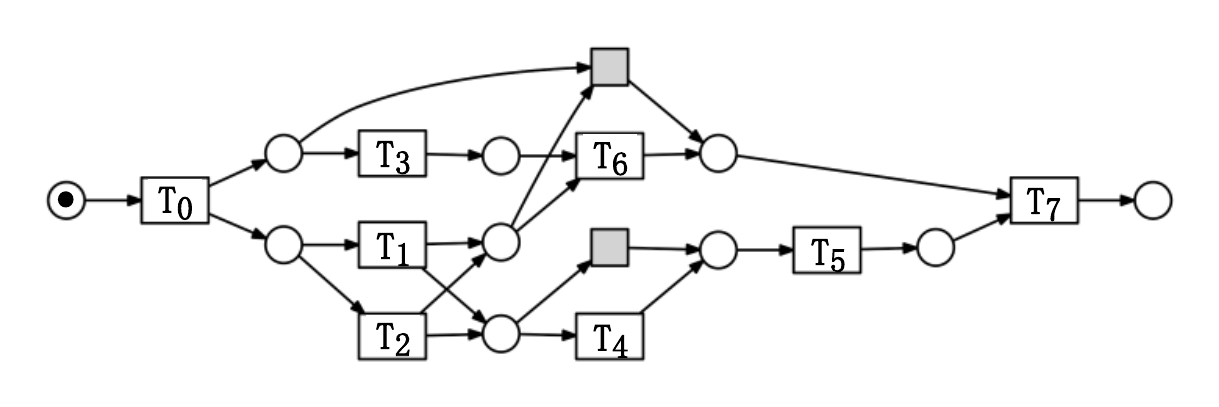}}\\
  \subfloat[][Behaviors after `soft' cut-off events can be consecutively continued.]{
    \label{fig:lseex2:b} 
    \includegraphics[width=5.4in]{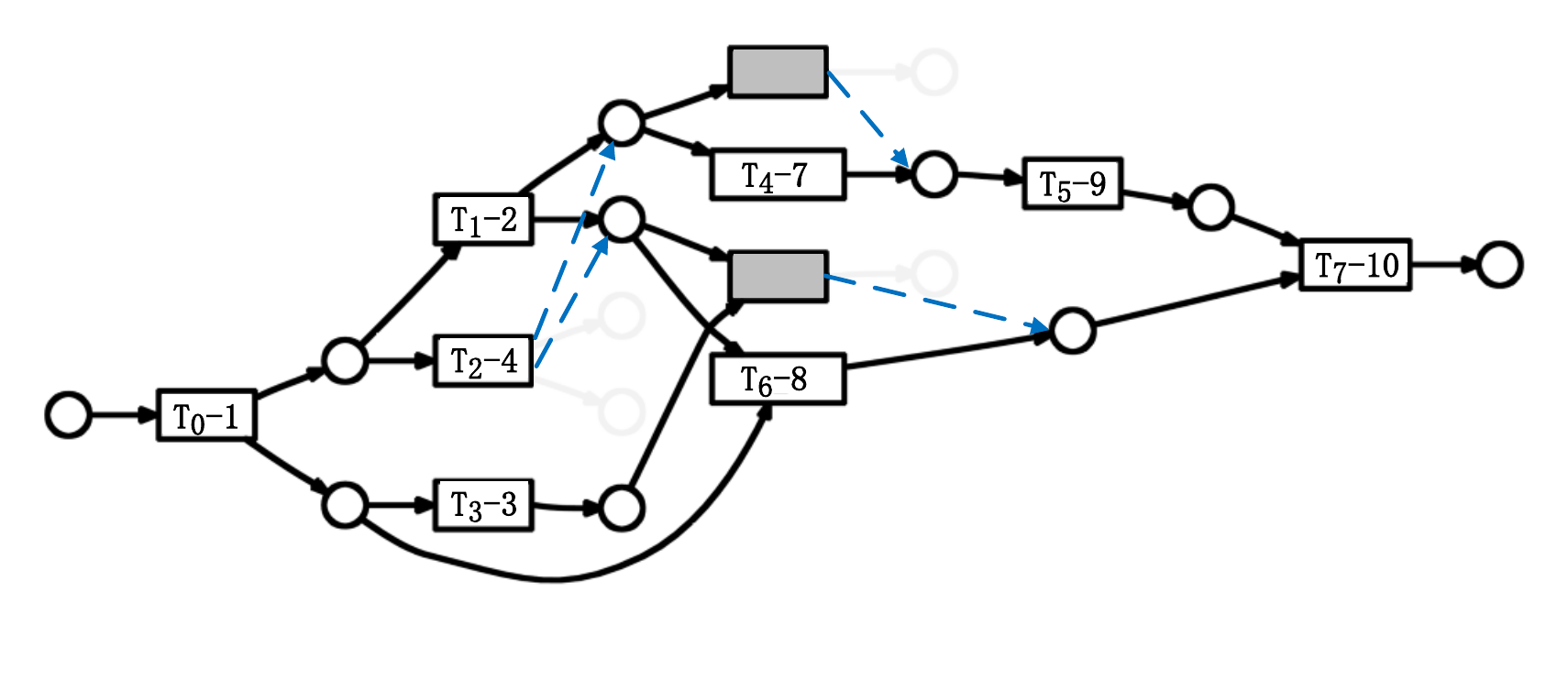}}
  \caption{Illustration of `soft' cut-off events}
\end{figure*}

Note that contextual isomorphism differs from the original isomorphism, i.e. the property mentioned and used by McMillan and Esparza in \cite{UNF95} and \cite{CPU02}. The previous isomorphism cannot guarantee the partial silent adjacency structures starting from $e$ can be immediately continued after $\Uparrow[e']$, since other cut-off events in $\Uparrow[e']$ might have prevented the unfolding of corresponding continued silent adjacency structures in the original CFP.

Compared to the previous isomorphism \cite{CPU02}, contextual isomorphism puts additional constraints on the co-sets `near' the cut-off events. With these additional constraints, it ensures that we can look for the isomorphic silent adjacency structures. Thus, less local extensions are required.

\begin{definition}[Redirected Silent Adjacency] \label{def:redirect} For an ordered pair of nodes $\langle n_1,n_2\rangle$ in a redirected CFP, they are in redirected silent adjacency, denoted as $n_1\triangleright'_s n_2$, if and only if one of the following conditions holds:
\begin{itemize}
\item[(1)] $n_1\in\bullet n_2$, or
\item[(2)] $n_2$ is a condition, $\exists e\in\bullet n_2\cup pre(n_2):n_1\triangleright_s' e$ and $h(e)$ is a silent transition, or
\item[(3)] $n_2$ is an event, and $\forall c\in \bullet n_2: n_1\triangleright_s' c\vee n_1$ co $c$.
\end{itemize}
\end{definition}

\begin{algorithm}
  \caption{Derive All pTARs from a Net System Using Redirected CFPs}\label{alg:redirected}
  \begin{algorithmic}[1]
    \label{alg:deriveaugstar}

    \Function {DerivePTARs}{$S,\pi_0$}
        \State $pTAR\leftarrow\emptyset$
        \State $Extended\leftarrow\emptyset$
        \State $\pi\leftarrow\pi_0$
        \State \textbf{for each} cut-off event $e_c\in\pi_0$ \textbf{do}
        \State \indent\textbf{if} $\not\exists e':e\sim e'$ \textbf{then} hard-cut-off$\leftarrow$ hard-cut-off $\cup \{e'\}$
        \State \indent\textbf{else for each} $e'_c\in\pi:e'_c\sim e_c$ \textbf{do}
        \State \indent\indent\indent\textbf{for each} $c\in e_c\bullet$ \textbf{do} $\pi\leftarrow\pi\cup\{\langle e_c,\gamma(c)\rangle\}$
        \State \textbf{for each} $e_1,e_2\in\pi: e_1\triangleright_s e_2$ \textbf{do}
        \State \indent $pTAR\leftarrow pTAR\cup\langle h(e_1),h(e_2)\rangle$
        \State \textbf{for each} $e_c\in$ hard-cut-off \textbf{do}
        \State \indent \textbf{for each} $e_1\in\pi:e_1\triangleright_s e_c$ \textbf{do}
        \State \indent \indent \textbf{if} $e_1\not\in Extended$ \textbf{then}
        \State \indent \indent\indent $Extended\leftarrow Extended\cup\{e_1\}$
        \State \indent \indent\indent $Fin_{e_1}=\Call{FinitePrefixForPTAR}{S,\pi,e_1}$
        \State \indent \indent\indent \textbf{for each} $ e_a,e_b :e_b\in Fin_{e_1}\backslash\pi,e_a\triangleright_s' e_b$ \textbf{do}
        \State \indent \indent\indent\indent $pTAR\leftarrow pTAR\cup\langle h(e_a),h(e_b)\rangle$
        \State \indent \indent\indent $\pi\leftarrow\pi\cup Fin_{e_1}$
        \State \indent \indent \textbf{end if}
        \State \textbf{return} $pTAR$
    \EndFunction
  \end{algorithmic}
\end{algorithm}

It is easy to see that the above definition of $\triangleright_s'$ is equivalent with the original $\triangleright_s$ within the original CFP $\pi$. The major difference lies in (2), for not only the original causal relations but also the redirected arcs represented by $pre(c)$ are taken into account while backtracking. In redirected CFPs, we need only replace the $\triangleright_s$ in Algorithm~\ref{alg:unfsilent},~\ref{alg:star} in the last section with $\triangleright_s'$ and we can get the pTAR derivation algorithms for redirected CFPs. So a pTAR computing algorithm using redirected CFP can now be given as Algorithm \ref{alg:redirected}, for which the original net $S$ its CFP $\pi_0$ are assumed to be given as input.

\begin{theorem} For two transitions $t_1,t_2$ $t_1<_{ptar}t_2$ if and only if there exists $e_1,e_2: h(e_1)=t_1,h(e_2)=t_2$ such that $e_1\in\pi, e_2\in\mu(\pi\backslash E_c\bullet)$ and $e_1\triangleright_{s}'e_2$.
\end{theorem}

\begin{proof} ($\Rightarrow$)By Theorem~\ref{thm:fincomplete}, for the original CFP $\pi$, there exists $e_1\in\pi, e'_2\in\beta_{e_1}\subseteq\mu(\pi)$ such that $h(e_1)=t_1,h(e_2)=t_2$ and $e_1\triangleright_{s}e_2$. Since in the original $\pi$, $\triangleright_s'$ and $\triangleright_s$ is equivalent, the case when $e_2\in\pi$ is trivial. Therefore we only focus on the cases when $[e_2]\backslash\pi\neq\emptyset$.

(1) Consider the set of events $E_{ext}=[e_2]\backslash\pi$ and its minimal conditions $C_{min}=Min(\bullet E_c)$. Without loss of generality, let $E_{ext}=\{e_{m_1},e_{m_2},...,e_{m_k},e_2\}$, where $e_{m_i}<e_{m_j}\rightarrow i<j$. Since $C_{min}\subseteq\pi\backslash E_c\bullet$, then since $\triangleright_s\subseteq\triangleright_s'$, we have $E_{ext}\subseteq\mu(\pi\backslash E_c\bullet)$ and therefore the proposition holds.

(2) If $C_{min}\not\subseteq\pi\backslash E_c\bullet$, then the pre-events of the conditions in $C_{min} \backslash (\pi\backslash E_c\bullet)$ are redirected due to contextual isomorphism. Given property (2) of context isomorphism, the new conditions they link to still form a co-set $C'_{min}$ isomorphic with $C_{min}$ and $\forall c\in C'_{min}:e_1\triangleright_s' c$. By isomorphism, a sequence of events $E_{ext}'=\{e'_{m_1},...\}$ ($|E_{ext}'|\geq 1$) isomorphic with a prefix of $E_{ext}=\{e_{m_1},e_{m_2},...,e_{m_k},e_2\}$ could be found after $C'_{min}$ until cut-off events are encountered, and it is easy to prove that $e_1\triangleright_s' e'_{m_i}$. If $E_{ext}'$ does not contain any cut-off event, then it will correspond to the full sequence of $E_{ext}$, and the proposition already holds. Otherwise, without loss of generality, suppose $E_{ext}'=\{e'_{m_1},...,e'_{m_p}\}$. Let $C''_{min}=(C'_{min}\backslash\bullet E_{ext}')\cup E_{ext}'\bullet$. If $C''_{min}\subseteq(\pi\backslash E_c\bullet)$, then by the same reasoning in Thm. \ref{thm:fincomplete} and (1), after $C''_{min}$, a sequence with the form of $\{e'_{m_{p+1}},...,e'_2\}, h(e'_2)=t_2$, can be found by performing local silent extensions such that $e_1\triangleright_s' e'_2, e'_2\in\mu(\pi\backslash E_c\bullet)$. If $C''_{min}\not\subseteq (\pi\backslash E_c\bullet)$, we could iterate (2) again from $C''_{min}$. Since $E_{ext}$ is finite, the iteration cannot be repeated infinitely, so we will eventually arrive at some $e_2'\in\mu(\pi\backslash E_c\bullet)$.

Let $E_p=\{e|e\in\bullet c,c\in C_{min}\}$. Since $\forall c\in C_{min}:e_1\triangleright_s c$, we know that $\forall e_p\in E_p:e_1\triangleright_s c$.

Let $E_r\subseteq\bullet C_{min}$ be the set of events that are redirected in the pre-events of the conditions in $C_{min}$.  Given property (2) of context equivalence, even if some pre-events of $C_{min}$ not included in $E_M$ are redirected in $\pi'$, the new conditions they link to still form a co-set $C'_{min}$ in $\pi'$ that is isomorphic with $C_{min}$ and $\forall c\in C'_{min}:e_1\triangleright c$. Starting from $C'_{min}$, we can find an event $e_2$ either through local silent extensions in $\mu(\pi')$ (by Proposition X) or as existing event in $\pi'$, such that $e'_2\in\mu(\pi')$ and $e_1\triangleright_s e'_2$

Since only the events in $E_M\cap\pi$ can be redirected, and only the events in $E_M\backslash\pi$ would affected, the case of $E_M\subseteq\pi$ is trivial. We now consider the case when $E_M\backslash\pi\neq\emptyset$.

(1) If none of the events in $E_M\cap\pi$ has been redirected in $\pi'$. Consider the set of conditions $C_{min}=Min(\bullet (E_M\backslash\pi))$ in the original $\pi$, and from $e_1\triangleright_s e_2$ we have $\forall c\in C_{min}:e_1\triangleright c$. Given property (2) of context equivalence, even if some pre-events of $C_{min}$ not included in $E_M$ are redirected in $\pi'$, the new conditions they link to still form a co-set $C'_{min}$ in $\pi'$ that is isomorphic with $C_{min}$ and $\forall c\in C'_{min}:e_1\triangleright c$. Starting from $C'_{min}$, we can find an event $e_2$ either through local silent extensions in $\mu(\pi')$ (by Proposition X) or as existing event in $\pi'$, such that $e'_2\in\mu(\pi')$ and $e_1\triangleright_s e'_2$.

(2) Suppose some $e_{m_i}$ in $E_M\cap\pi$ has been redirected in $\pi'$. Then since either $e_{m_i}<e_{m_{i+1}}$ or $e_{m_i}$ co $e_{m_{i+1}}$, by property (2) of Def. X, we know that in the redirect CFP $\pi'$, some event $e_{m_{i+1}}': h(e_{m_{i+1}}')=h(e_{m_{i+1}})$ and obviously $e_1\triangleright_s e_{m_{i+1}}'$. If $e_{m_{i+1}}'$ is not a cut-off event, then we can find in $\pi'$ an $e_{m_{i+2}}'$ which is isomorphic with $e_{m_{i+2}}$ and $e_1\triangleright_s e_{m_{i+2}}'$. This can be iterated until we reach some event isomorphic with $e_2$, which proves the proposition, or some $e_{m_j}'$ that is a cut-off event. If $e_{m_j}'$ is not redirected, then by the property of contextual equivalence and Prop. X, an event $e'_2$ after $e_{m_j}'$ can be discovered through local silent extensions such that $e'_2\in\mu(\pi')$ and $e_1\triangleright_s e'_2$. If $e_{m_j}'$ is redirected, then we can iterate the process of (2). This iteration cannot continue infinitely because $E_M$ is finite, and eventually we will arrive at $\exists e'_2\in\mu(\pi')$ and $e_1\triangleright_s e'_2, h(e_1)=t_1$ and $h(e'_2)=t_2$.

$\Leftarrow$) Given that $e_1\triangleright_s' e_2$, consider the co-set $B_0=\bullet e'_2$. As $\forall c\in\bullet e_2:c$ co $e_1$ or $e_1\triangleright_s' c$, let $B_{0_{co}}=\{c\in B_0$ co $e_1\}$, $B_{0_{prec}}=\{c\in B_0|e_1\triangleright_s' c\}$. Obviously $B_0=B_{0_{co}}\cup B_{0_{prec}}$.

Take any condition $c'$ from $B_{0_{prec}}$, since $e_1\triangleright_s' c'$, it follows that $\exists e_k\in\bullet c'\cup prec(c'): e_1\triangleright_s' e_k\}$ and $h(e_k)$ is a silent transition.

If $e_k<c'$, then let $B_1=(B_0\backslash e_k\bullet)\cup \bullet e_k$. We have $e_k$ co $B_0\backslash e_k\bullet$ (obviously $e_k<c'$ cannot hold, and for any $c'\in B_0\backslash e_k\bullet$, if $e_k\#c\vee c<e_k$, it will lead to a contradiction with the fact that $B_0$ is a co-set). We can prove that $\bullet e_k$ are in $co$ with $B_0\backslash e_k\bullet$ (for any $c''\in \bullet e,c'\in B_0\backslash e_k\bullet$, any one of $c''\#c'$, $c''<c'$ or $c'<c''$ would lead to the contradiction with the fact that $e_k$ co $c'$) and consequently $B_1=(B_0\backslash e_k\bullet)\cup \bullet e_k$ is a co-set.

If $e_k\in prec(c')$ then let $B_1=(B_0\backslash pre^{-1}(e_k))\cup \bullet e_k$. In the original $\pi$, by the property (2) of Def. \ref{def:context}, as $e_k\bullet$ have isomorphic context with $pre^{-1}(e_k)$, we can similarly prove that $e_k$ are in $co$ with $B_0\backslash pre^{-1}(e_k)$. Consequently, $\bullet e_k$ are in $co$ with $B_0\backslash pre^{-1}(e_k)$ and therefore $B_1=(B_0\backslash pre^{-1}(e_k))\cup \bullet e_k$ is a co-set. In both cases $B_1$ is a co-set Similarly, we could have $B_{1_{co}}=\{c\in B_1$ co $e_1\}$, $B_{1_{prec}}=\{c\in B_1|e_1\triangleright_s' c\}$.

In both cases, we could generate from $B_0$ a new co-set $B_1$ such that $e_k$ is enabled at some reachable marking $M_k$ that contains $h(B_1)$, and the firing of $h(e_k)$ from that reachable marking will lead to a reachable marking $M_{k+1}$ that contains $B_0$ and therefore enables $h(e_2)$, which gives $M_{k}\xrightarrow{h(e_k)}M_{k+1}$, and $h(e_2)$ is enabled at $M_{k+1}$.

Again, since $e_1\triangleright_s' e_k$, we have $\forall c\in \bullet e_k:c$ co $e_1$ or $e_1\triangleright_s' c$. Therefore, as $B_1\backslash B_0=\bullet e_k$, we have $\forall c\in \bullet B_1:c$ co $e_1$ or $e_1\triangleright_s' c$. Again, we can partition $B_1$ as $B_{1_{co}}=\{c\in B_1$ co $e_1\}$ and $B_{1_{prec}}=\{c\in B_1|e_1\triangleright_s' c\}$ ($B_1=B_{1_{co}}\cup B_{1_{prec}}$). Take any condition $c$ from $B_{1_{prec}}$, and we can repeat the above iterations to find the next event $e_{k-1}$ such that $h(e_{k-1})$ is a silent transition and $M_{k-1}\xrightarrow{h(e_{k-1})}M_{k}\xrightarrow{h(e_k)}M_{k+1}$. Since $\mu(\pi)$ is finite, we could arrive with finite steps at $e_1$ and $M_1\xrightarrow{h(e_1)}...M_{k}\xrightarrow{h(e_k)}M_{k+1}$ and $h(e_2)$ is enabled at $M_{k+1}$. Therefore we have $h(e_1)<_{ptar}h(e_2)$.

In case (1) let $B_0'=\{c\in B_0|c\in e_{m_k}\bullet\}$. Since $e_{m_k}\triangleright_s' e_2$ and $e_{m_i}<e_{m_j}\rightarrow i<j$, we can easily prove that $B_{k-1}=(B_k\backslash B'_k)\cup \bullet e_{m_k}$ is a co-set, and $h(e_{m_k})$ is `enabled' at any reachable marking $M_k$ containing $h(B_{k-1})$. And since $B'_k\subseteq e_{m_k}\bullet$, the firing $h(e_{m_k})$

In case (2) let $B_0'=\{c\in B_0|e_{m_k}\in pre(c)\}$. By property (2) of Def. X, we can also easily prove that $B_{k-1}=(B_k\backslash B'_k)\cup \bullet e_{m_k}$ is a co-set, and $h(e_k)$ is `enabled' at any reachable marking containing $h(B_{k-1})$.

\end{proof}

In Algorithm~\ref{alg:redirected}, since we can re-use existing CFP structures to observe behaviors after soft cut-off events, significant amount of computational costs can be reduced, as will be demonstrated with detailed performance analysis in the experiment section.

\section{Performance Evaluation}\label{sec:Eval}

\subsection{TAR Computing Performance}
We conduct experiments to evaluate the effectiveness of our proposed CFP-based algorithms using both industrial (see \cite{ISC09} for further details) and artificial process model datasets. We have implemented our algorithms based on the open source BPM tool jBPT\footnote{http://code.google.com/p/jbpt} and all our programs and the process model datasets have been made available on-line\footnote{http://laudms.thss.tsinghua.edu.cn/trac/Test/raw-attachment/wiki/Share/CTAR.zip}. All experiments were run on a PC with Intel Dual Core I5 CPU@2.8G, 4G DDR3@1333MHz, and Windows 7 Enterprise OS. For simplicity, we focus on the TAR computing performance in this experiment, as many existing algorithms do not support pTAR computation. But still, this could demonstrate the performance improvement of our accelerated method (IMPROVE) against the baseline coverability based method (GENERAL) and other existing methods, since clearly the improvement in pTAR computing performance, if there is any, will only be greater than that of TAR.
\subsubsection{Real-world Data}\label{subsec:RsnEva}
Table 2 summarizes the features of our experiment datasets \cite{ISC09} consisting of 5 libraries, from which we extract all the bounded models for evaluation. The table also shows the degree of concurrency found in the process models, i.e. the maximum and average number of tokens that occur in a single reachable non-error state of the process. Since RG-based algorithms may encounter state-explosions when applied on nets with high concurrencies, we count the target net as ``Time Limit Exceeded for RG'' if it fails to construct a RG within acceptable time limit, as shown in the last row of the table.

\begin{table*}
\centering
\caption{Features of the five process model libraries used in experiments}
\label{tbl:RsnEva}
\begin{tabular}{|c|c|c|c|c|c|}
\hline
& \textbf{A2} &\textbf{B2} & \textbf{B3} & \textbf{B4} & \textbf{M1}\\ \hline
\texttt{Num. of Bounded Nets}&256&392&338&272&27 \\\hline
\texttt{Avg. Transitions/Places}&28/44&29/43&27/41&28/41&48/45 \\\hline
\texttt{Avg./Max Concurrency}&2/13&8/14&16/66&14/33&2/4 \\\hline
\texttt{Time-Limit Exceeded for RG}&0&39&36&24&3 \\\hline
\end{tabular}
\end{table*}


First, we compare the time costs (see Fig. \ref{fig:perf:a}) of our algorithm GENERAL (Algorithm~\ref{alg:basic}) which is based on coverability reduction against Jin's approach (JIN) based on the primitive unfolding derivation rules in \cite{ER12} (which does not support arbitrary bounded nets since it assumes free-choice WF-nets), and two RG-based algorithms (pure RG-based algorithm in \cite{pTAR12} (RG) and Zha's improved RG based approach (ZHA) \cite{TAR10}). Time costs of building CFPs of unfoldings alone (CFP-BUILD) in each library are also shown in Fig. \ref{fig:perf:a}. It can be observed that as nets in B2, B3 and B4 contain more concurrent structures, GENERAL performs significantly better than all RG-based algorithms while only being a bit slower than RG in A2, M1, which can be explained by the lacking of concurrent structures in these two libraries. By comparing their overall time cost (TOTAL), we observe that GENERAL is around 11.98 times and 9.78 times faster than RG and ZHA respectively. We also see that whereas JIN only supports sound and free-choice nets, our GENERAL achieves a much greater generality (supports any bounded net system) at the expense of only a bit more overhead and shows comparable performance with JIN (only 1.1\% more time cost in total).

\begin{figure}
\centering
  \subfloat[][Time cost comparison on bounded nets]{
    \label{fig:perf:a} 
  \includegraphics[width=3.0in]{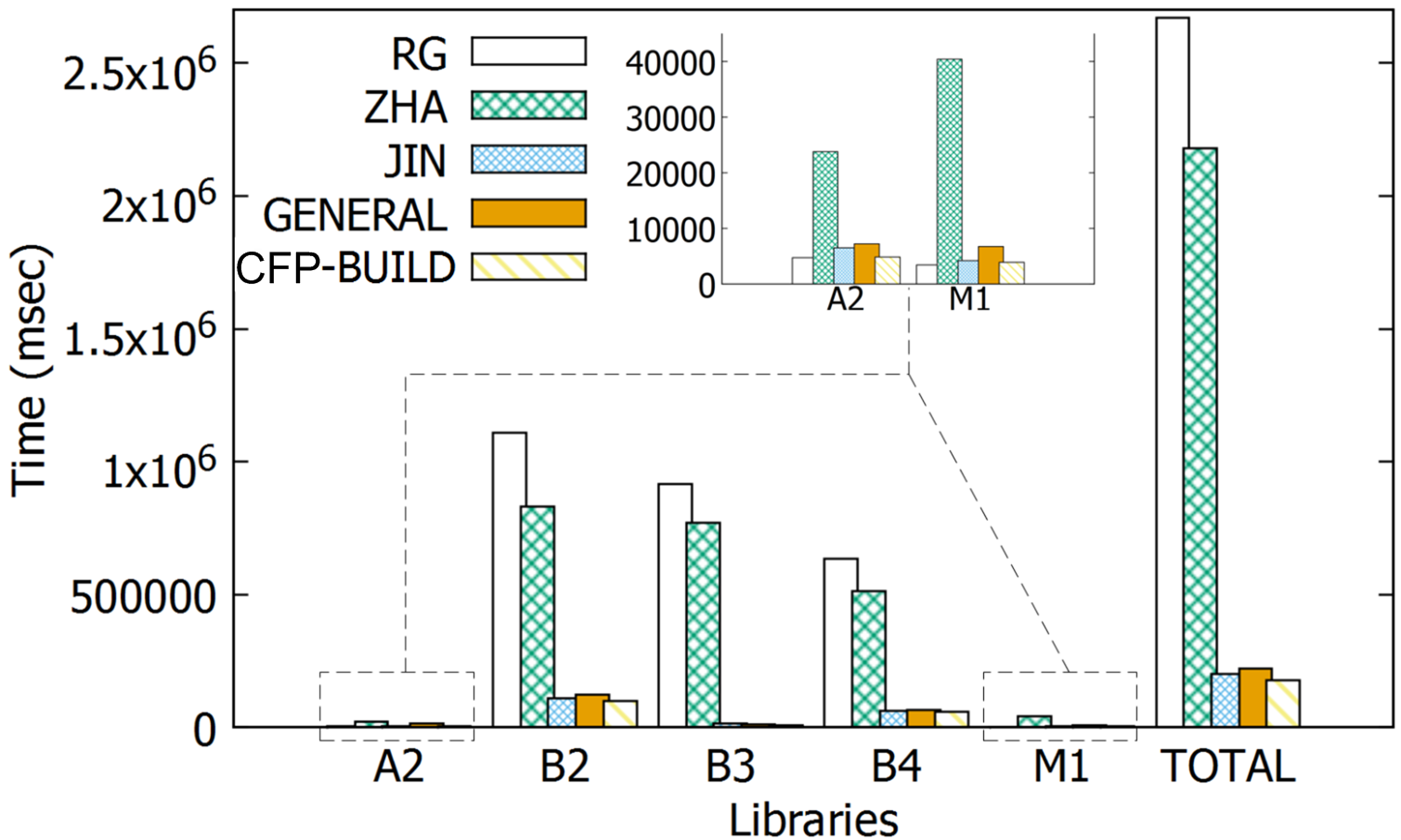}}
  \hspace{0.03in}
  \subfloat[][Overhead of GENERAL \& IMPROVED]{
    \label{fig:perf:b} 
    \includegraphics[width=3.2in]{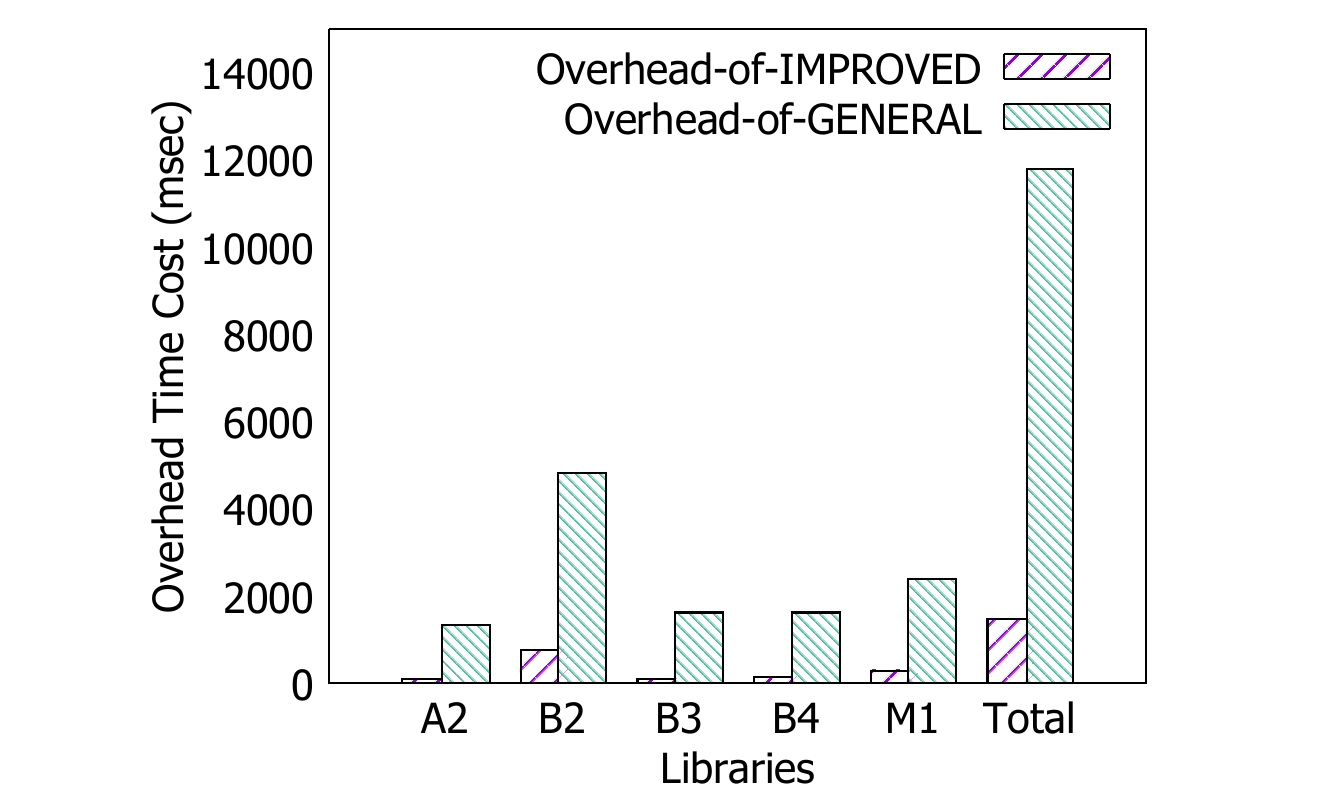}}
  \caption{Efficiency evaluation of proposed TAR derivation algorithms}
  \label{fig:perf} 
\end{figure}

Furthermore, we observe the acceleration effect of the Algorithm~\ref{alg:final} (IMPROVED), which exploits the relations in CFP further than GENERAL  (Fig. \ref{fig:perf:b}). We compute the overhead of IMPROVED and GENERAL, i.e., their total TAR derivation time cost minus CFP-BUILD time, which costs the same for both algorithms. In Fig. \ref{fig:perf:b}, the total overheads for GENERAL to derive TAR from CFPs are around 7 times larger than IMPROVED, which indicates that exploiting more CFP relations results in significant acceleration.


\subsubsection{Scalability Evaluation}\label{subsec:PfmEva}
We compare the scalability of GENERAL, IMPROVED with existing JIN, ZHA and RG on artificial Petri nets with growing concurrency structures. In Test A, each net is a single parallel structure with exactly one transition on each branch, and we observe how the increase in the number of concurrent branches affects TAR derivation time. In Test B, we generate concurrent branches of a fixed number (i.e. 5) and observe the effect of the number of transitions on each branch on the performance of algorithms.

\begin{figure}
  \centering
  \subfloat[][Test-A: Breadth-Growth]{
    \label{fig:scala:b} 
    \includegraphics[width=2.7in]{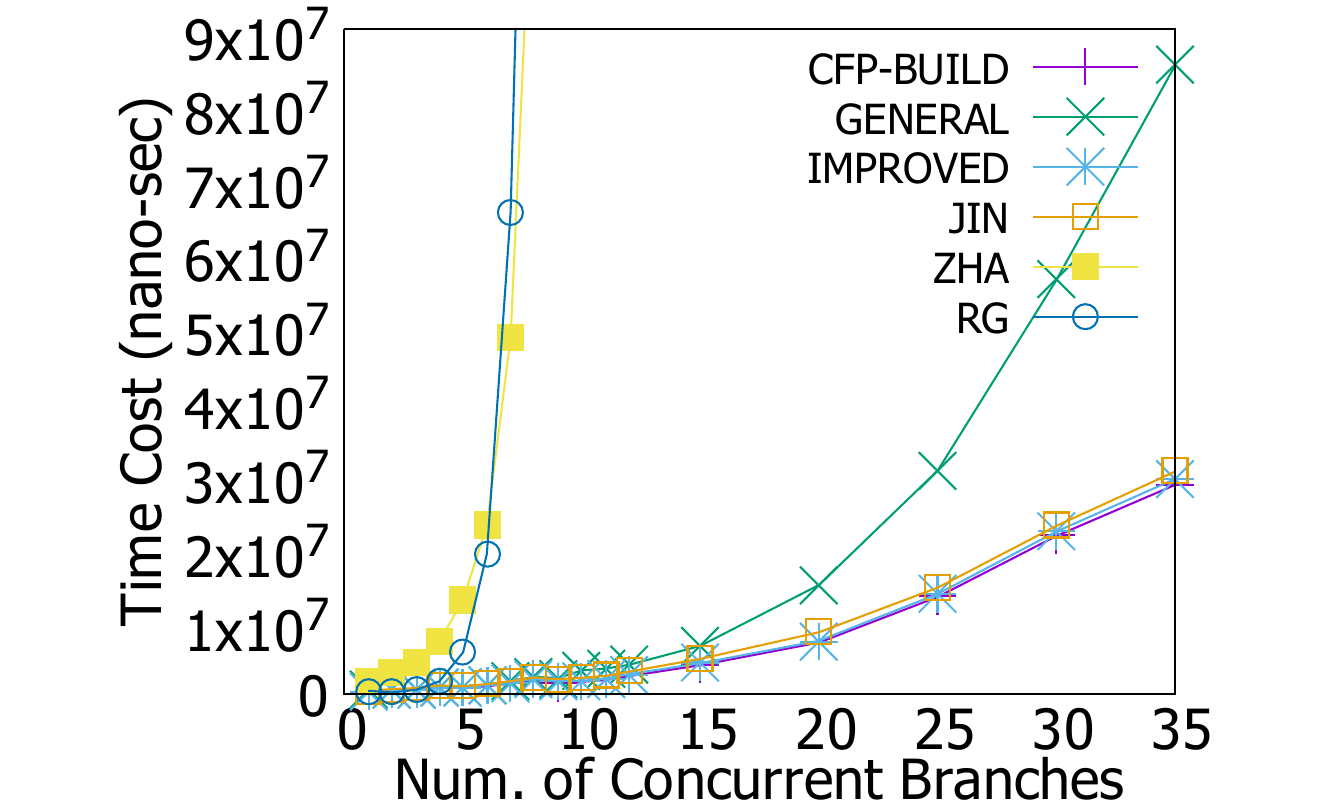}}
  \hspace{0.03in}
  \subfloat[][Test-B: Depth-Growth]{
    \label{fig:scala:a} 
    \includegraphics[width=2.7in]{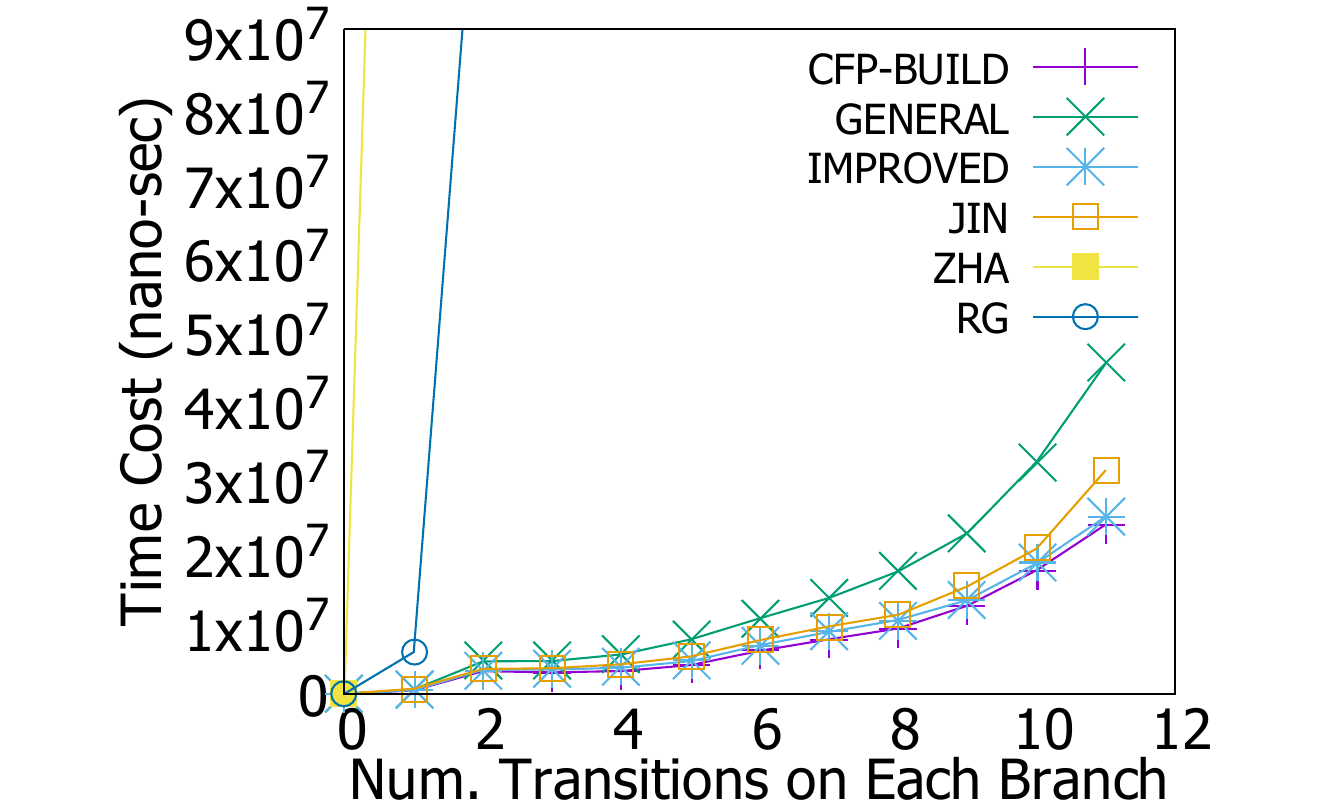}}
  \label{fig:subfigjump} 
  \caption{Scalability evaluation of TAR derivation algorithms}
\end{figure}

In both tests, CFP-based algorithms (GENERAL, IMPROVED, JIN) scale much better than RG-based algorithms (RG, ZHA). But JIN does not support non-free-choice nets or pTAR whereas GENERAL and IMPROVED do. Among CFP-based algorithms, IMPROVED scales the best with the growing size of concurrency structures, which is more efficient than GENERAL in both tests and has comparable or even better performance than JIN.

\subsection{pTAR Computing Performance}

To meet their needs for analysis, users may treat various sets of transitions as `projected' transitions (and the rest silent transitions). Therefore, the percentage of silent transitions may differ among scenarios. To observe the performance of pTAR derivation algorithms under this diversity, we run our algorithms on data-sets with different percentages of silent tasks. These data-sets are prepared by marking the same IBM industrial process data-set \cite{ISC09} with varying percentages of silent transitions. We prepare data-sets with five levels of silent task percentage, namely, 10, 30, 50, 70 and 90 percent. Fig.~\ref{fig:ptarperf:a} shows the total time cost to derive pTAR from these data-sets under varying percentages of silent transitions, using the three algorithms we have discussed. In the next figures, RG stands for the original reachability-graph-based pTAR computation algorithm \cite{pTAR12}, CU and RCU denote our algorithm based on locally continued unfolding and its accelerated version using redirected CFP. Fig.~\ref{fig:ptarperf:b}, ~\ref{fig:ptarperf:c} and ~\ref{fig:ptarperf:d} shows the detailed time cost spent on each library by the three algorithms respectively.

\begin{figure}
  \centering
  \subfloat[][Total Time Cost of Each Method]{
    \label{fig:ptarperf:a} 
    \includegraphics[width=2.7in]{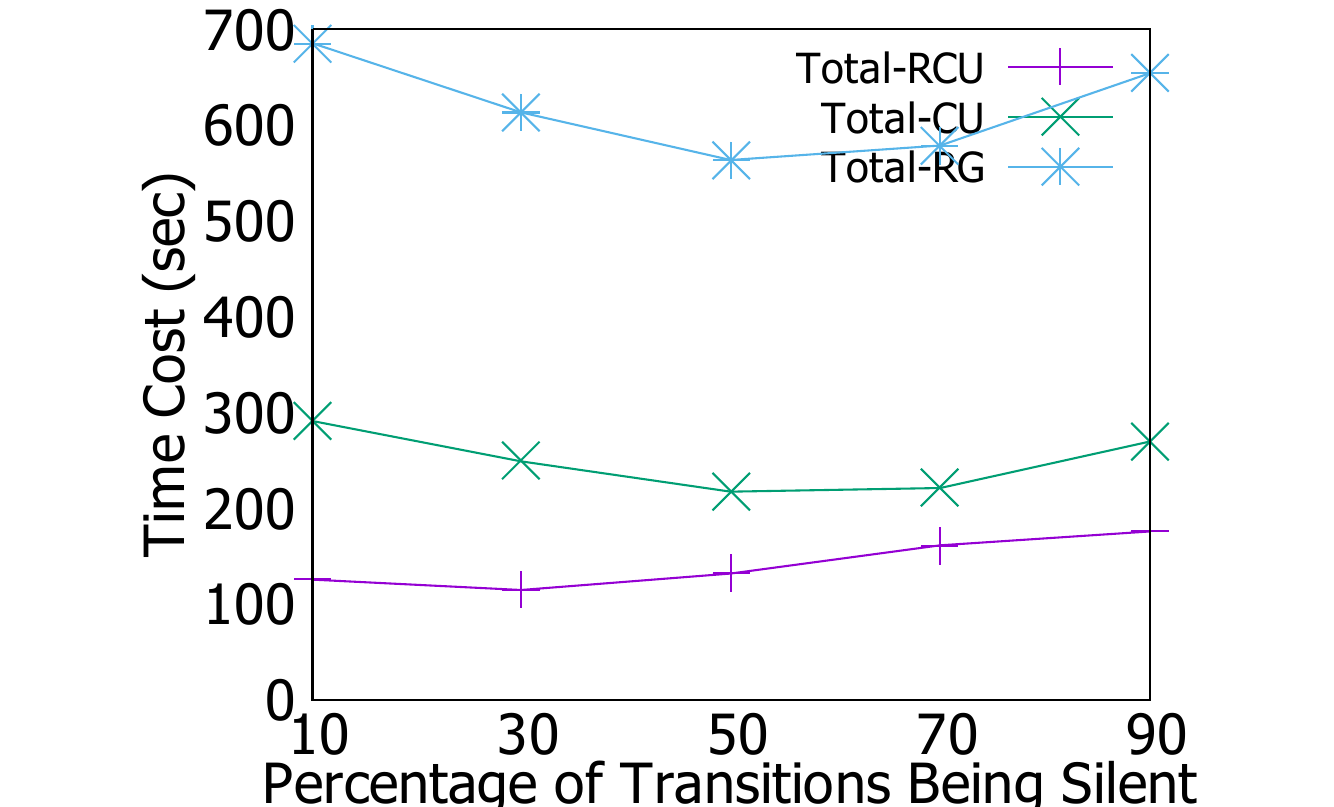}}
  \hspace{0.1in}
  \subfloat[][Detailed Performance of RCU]{
    \label{fig:ptarperf:b} 
    \includegraphics[width=2.7in]{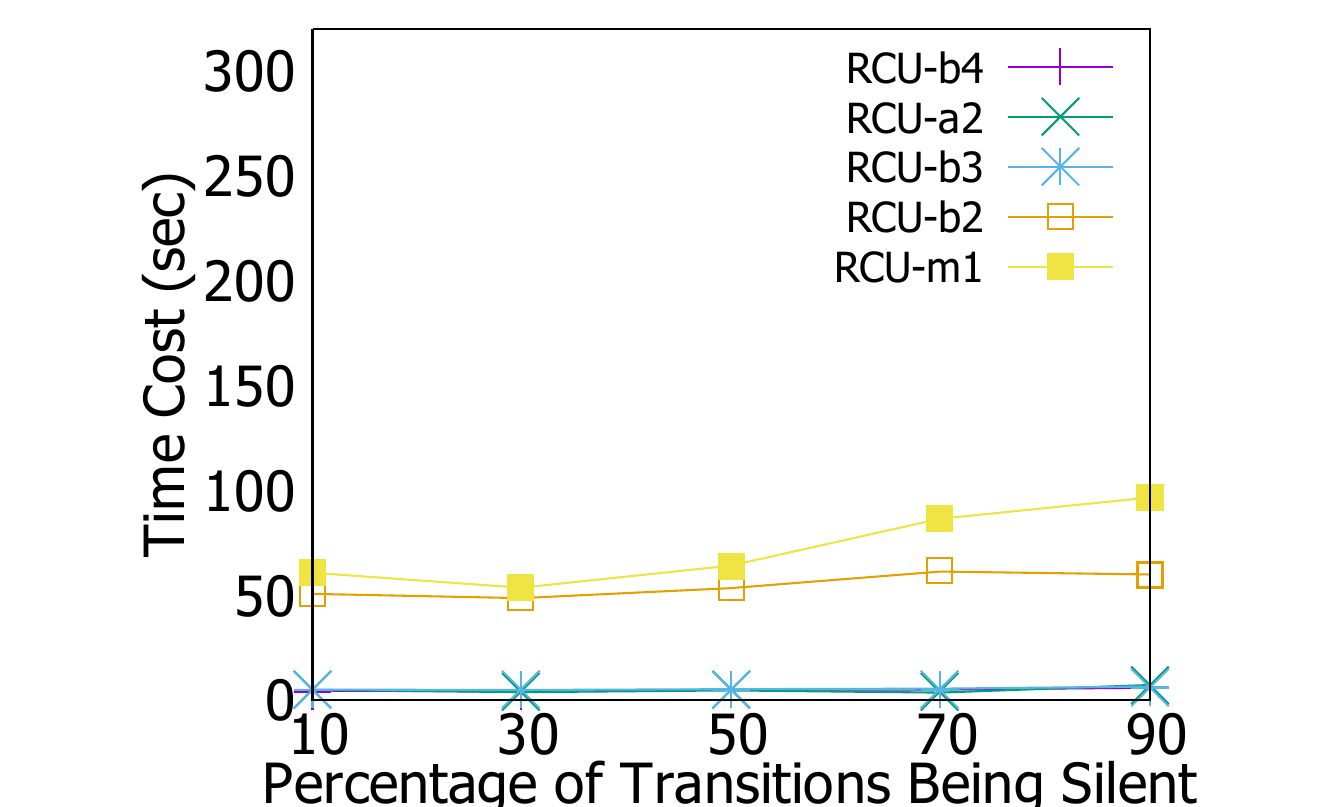}}
  \hspace{0.1in}
  \subfloat[][Detailed Performance of CU]{
    \label{fig:ptarperf:c} 
    \includegraphics[width=2.7in]{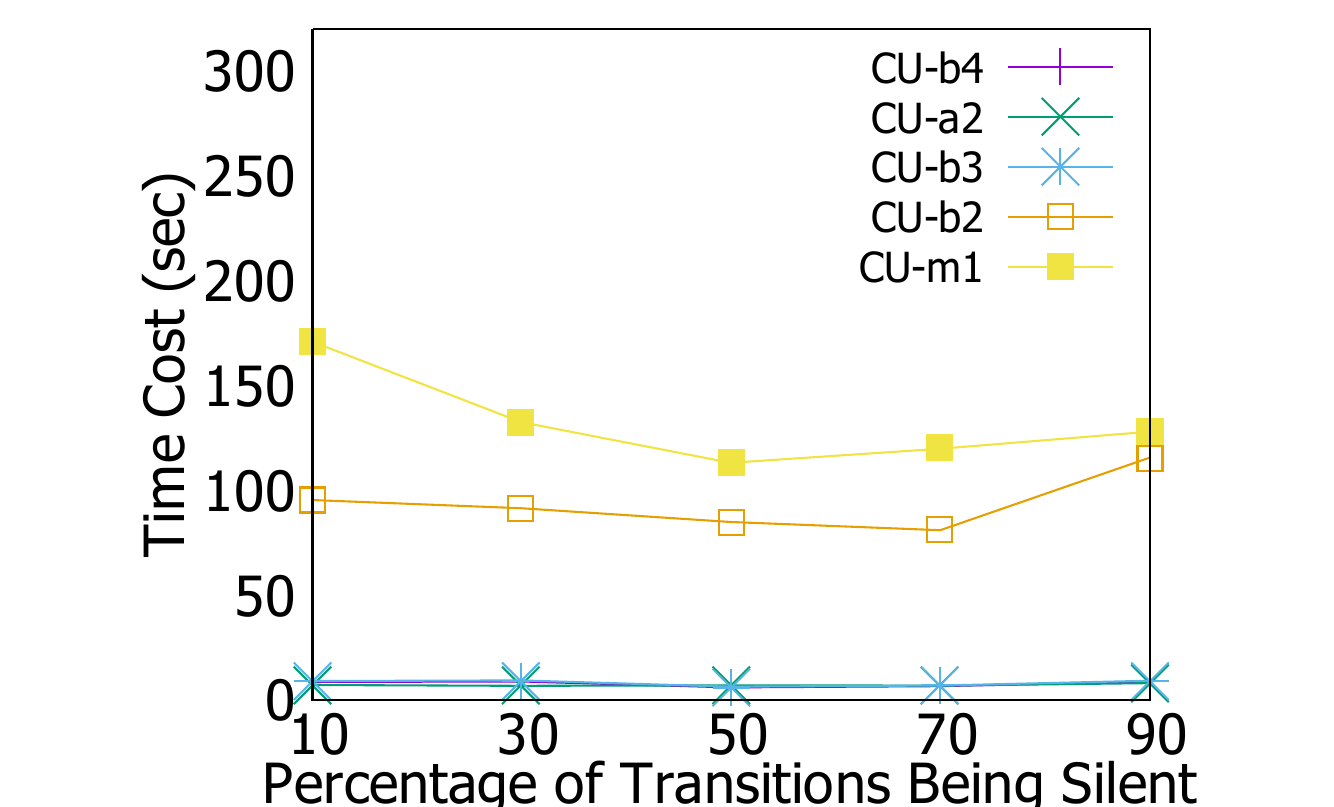}}
  \hspace{0.1in}
  \subfloat[][Detailed Performance of RG]{
    \label{fig:ptarperf:d} 
    \includegraphics[width=2.7in]{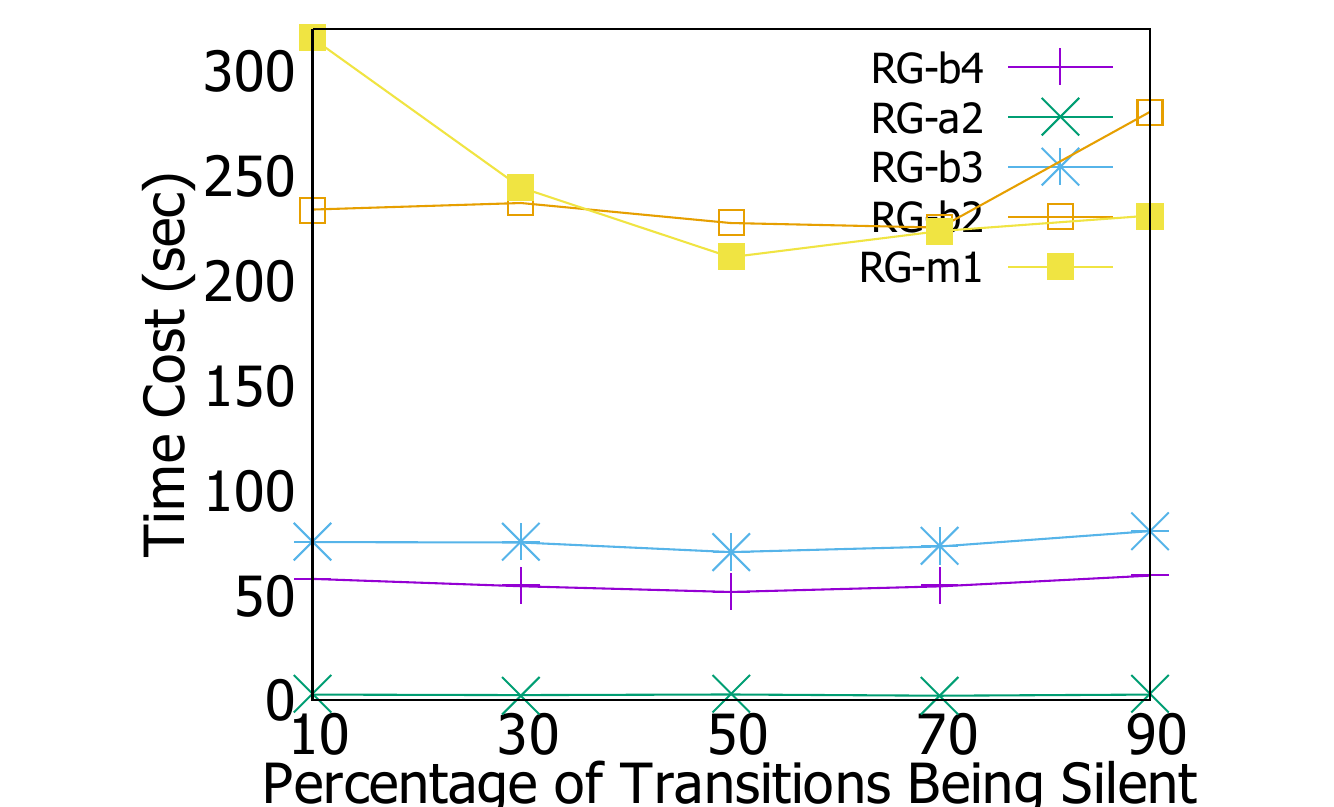}}

  \label{fig:ptarperfo} 
  \caption{pTAR time cost under varying percentage of silent transitions}
\end{figure}

As can be observed in the above figures, RCU performs the best among all three algorithms we have tested (RG, CU and its accelerated version RCU). Moreover, we could see that by using the redirected CFP and avoid unnecessary local silent extensions, the time cost of RCU is significantly less than CU.

\section{Combining TAR and BP for Similarity-based Process Model Querying}\label{sec:Comb}

The implications of the above results are not only about the improved efficiency of TAR computation. As Behavior Profiles (BP) are also efficiently computable through Petri net unfolding~\cite{CBP10}, we can easily enhance BP information with additional TAR information by re-using the information contained in the unfolding structure (and vice versa). Therefore, the long-expected combined usage of TAR and BP for process model similarity 
estimation~\cite{CII12} could now be realized with low overhead because of the above efficient algorithms. 

In Sect.~\ref{sect:intro}, we have demonstrated the advantage of TAR over BP in terms of the discriminative power for various process model behaviors (i.e. loops, parallel branches and silent transitions). In this section, we will conduct experiments using industrial benchmark datasets to further illustrate the implications of our proposal for process model similarity estimation.

Our evaluation is based on the scenario of process model retrieval. The problem of  process model retrieval is defined as: given some query process model $p$ and a repository of process models $\mathcal{C}$, find in $\mathcal{C}$ process models that are most similar with $p$, and rank them in order of similarity. For benchmark dataset, we used the SAP reference model repository as in~\cite{IS11}, which contains 100 process models in total, among which there are 10 query models. The dataset comes with ground-truth tags provided by experts, suggesting which models are similar to which query models by boolean values. For the similarity estimation between tasks with different activity descriptions, we used the label edit distance threshold as in~\cite{IS11} to determine whether two tasks with different descriptions points to the same activity.

For similarity estimation, we use the classic definition of Jaccard coefficient, by which similarity from the perspective of TAR and BP can be given respectively as follows.

\begin{definition}[TAR Similarity]For net system $N_1=(P_1, T_1, F_1)$ and $N_2=(P_2, T_2, F_2)$ (with initial markings of $M_1,M_2$ respectively), let
$TS_1$ and $TS_2$ be the TAR set of the two nets, then the TAR similarity between $N_1$ and $N_2$ can be given as follows:
$$
sim_{T}((N_1,M_1),(N_2,M_2))=\frac{|TS_1\cap TS_2|}{|TS_1\cup TS_2|}
$$
\end{definition}

\begin{definition}[BP Similarity]For net system $N_1=(P_1, T_1, F_1)$ and $N_2=(P_2, T_2, F_2)$(with initial markings of  $M_1,M_2$ respectively), $\mathcal{B}_1$ and $\mathcal{B}_2$ be the BP set of the two nets, then the BP similarity between $N_1$ and $N_2$ can be given as follows:
$$
sim_{BP}((N_1,M_1),(N_2,M_2))=\frac{|\mathcal{B}_1\cap \mathcal{B}_2|}{|\mathcal{B}_1\cup \mathcal{B}_2|}
$$
\end{definition}

Consequently, the ranking functions for process models in $C$ given a query model $p\in P$ can be given as follows

\begin{definition}[TAR Ranking] For some query model $p_p\in P$ and some target model $p_c$ in repository $C$, the TAR ranking of $p_c$ is given as $$R_{T}(p_p,p_c)=|\{p_c'\in C|sim_{T}(p_p,p'_c)\geq sim_{T}(p_p,p_c)\}|$$
\end{definition}

\begin{definition}[BP Ranking] Similarly, for some query model $p_p\in P$ and some target model $p_c$ in repository $C$, the BP ranking of $p_c$ is given as $$R_{BP}(p_p,p_c)=|\{p_c'\in C|sim_{BP}(p_p,p'_c)\geq sim_{BP}(p_p,p_c)\}|$$
\end{definition}

We may use the product of TAR ranking and BP ranking to combine their influence to the similarity ranking result as follows：
\begin{definition}[TAR-BP Ranking] For some query model $p_p\in P$ and some target model in $p_c\in C$, the TAR-BP ranking of $p_c$ can be given as $$R_{TBP}(p_p,p_c)=|\{p_c'\in C|R_{T}(p_p,p'_c)\cdot R_{BP}(p_p,p'_c)\leq$$$$R_{T}(p_p,p_c)\cdot R_{BP}(p_p,p_c)\}|$$
\end{definition}

With the above similarity ranking functions, we evaluate the process model retrieval performance when (1) only TAR (2) only BP and (2) both TAR and BP are considered. In Fig. \ref{fig:precision-recall} we compared the precision-recall performance of these three cases. It can be observed that the combined usage of TAR and BP could improve the overall performance of process model retrieval, which agrees with our previous discussions in Sect. \ref{sect:intro}. 

\begin{figure}
  \centering
 \includegraphics[width=2.8in]{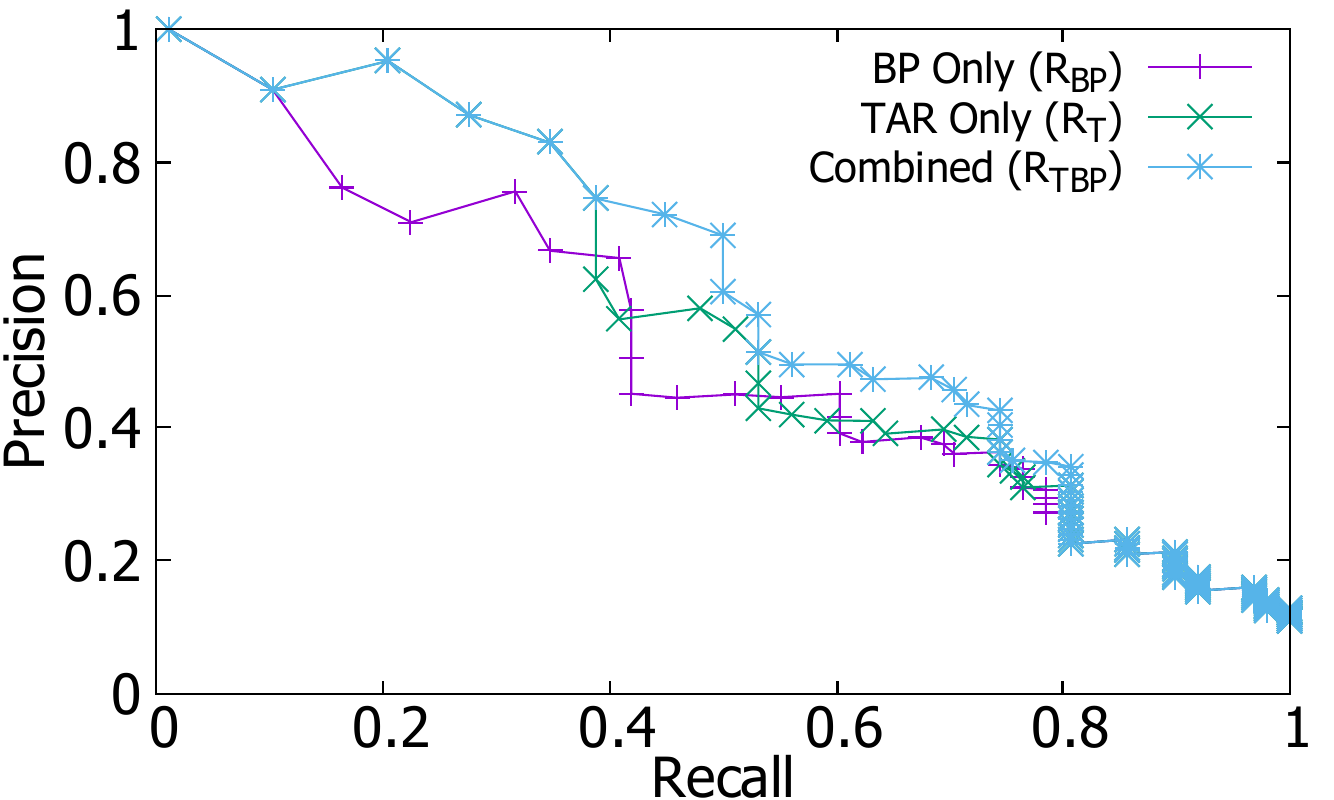}
    \caption{Recall-precision performance of different similarity measures}
  \label{fig:precision-recall}
\end{figure}

\section{Conclusion}\label{sec:Conc}

In this article we proposed comprehensive strategies for the computation of TAR and pTAR based on Petri net unfoldings. These strategies attempt to exploit as much information in unfoldings as possible, by translating causal patterns and $co$-relation patterns into TAR/pTAR results. The fundamental challenge of unfolding based TAR/pTAR computation is the handling of cut-off events, which may truncate information needed for TAR/pTAR confirmation. Novel techniques to derive locally continued unfolding and redirected CFP for efficient recovery of necessary information after cut-off events are proposed, proved and evaluated in this article. Experiments show that the proposed strategies achieve significant performance improvement over existing methods based on reachability graphs. This
will support further researches and more scalable applications of TAR/pTAR in various business process analytic tasks.

\ifCLASSOPTIONcaptionsoff
  \newpage
\fi



%

%
 

\begin{IEEEbiography}[{\includegraphics[width=1in,height=1.25in,clip,keepaspectratio]{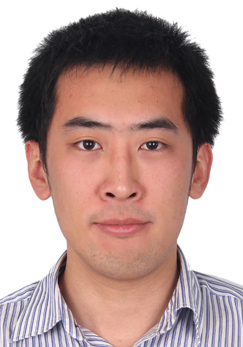}}]{Jisheng Pei} received his BS degree in Computer Software and his PhD degree in Computer Science from Tsinghua University. His research interests include the processing and analysis methods of data provenance, process mining, similarity measurements of process traces and models, as well as behavior analytics and service computing.
\end{IEEEbiography}

\begin{IEEEbiography}[{\includegraphics[width=1in,height=1.25in,clip,keepaspectratio]{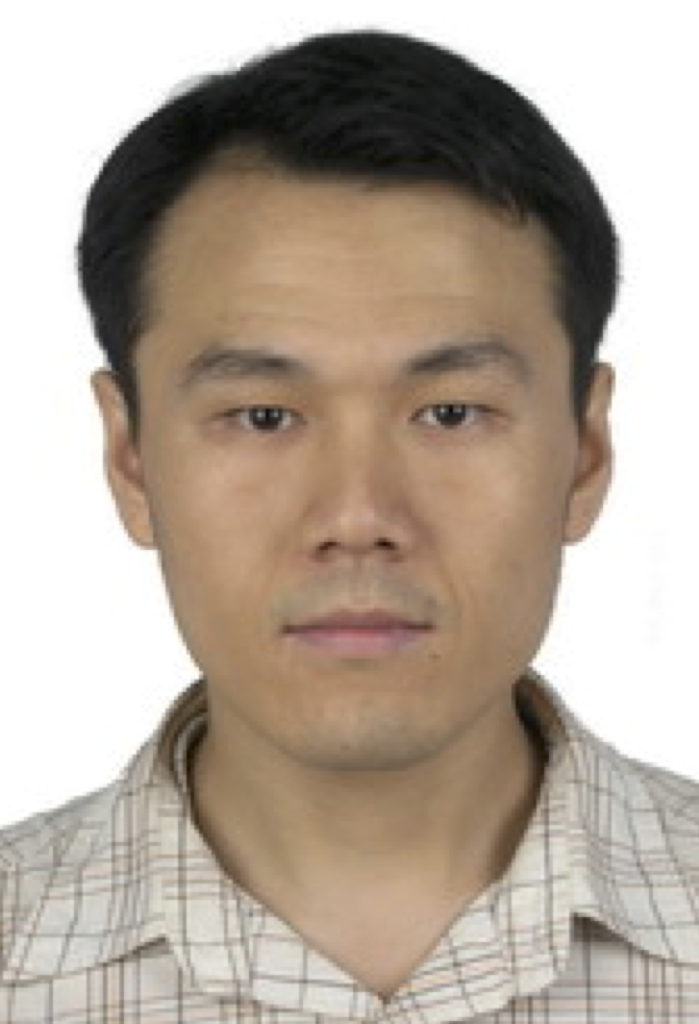}}]{Lijie Wen}
received the BS degree, the Ph.D. degree in Computer Science and Technology from Tsinghua University, Beijing, China, in 2000 and 
2007 respectively. He is currently an associate professor at School of Software, Tsinghua University. His research interests are focused on process mining, process data management (e.g., log completeness, trace clustering, process similarity, process indexing and retrieval), and lifecycle management of workflow for big data analysis.
\end{IEEEbiography}

\begin{IEEEbiography}[{\includegraphics[width=1in,height=1.25in,clip,keepaspectratio]{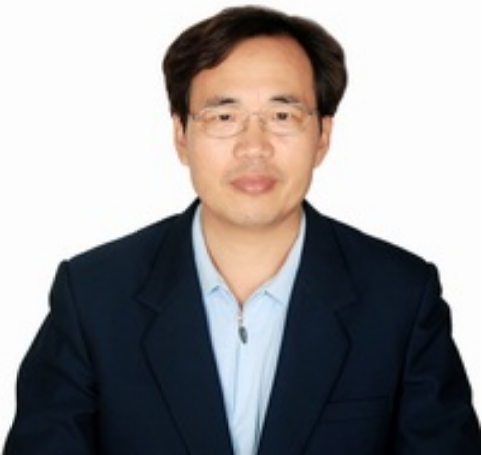}}]{Xiaojun Ye}
received the BS degree in mechanical engineering from Northwest Polytechnical University, Xian, China, in 1987 and the Ph.D. degree in information engineering from INSA Lyon, France, in 1994. Currently, he is a professor at School of Software, Tsinghua University, Beijing, China. His research interests include cloud data management, data security and privacy, and database system testing.
\end{IEEEbiography}

\begin{IEEEbiography}
[{\includegraphics[width=1in,height=1.25in,clip,keepaspectratio]{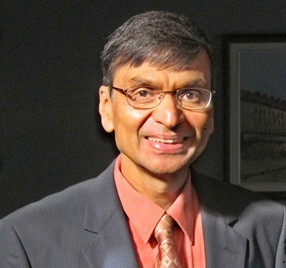}}]{Akhil Kumar}
has a PhD from the University of California, Berkeley, and is a full professor of information systems at the Smeal College of Business, Pennsylvania State University. His current research interests are in business process management (BPM) and workflow systems, process mining, supply chain and business analytics, and health IT. He is a senior member of IEEE and an associate editor for ACM Transactions on Management Information System.
He served as a co-chair for the International Business Process Management Conference in 2017. He has been a principal investigator for National Science Foundation and also received support from IBM, HP, and other organizations for his work.
\end{IEEEbiography}


%




\end{document}